\documentclass[a4paper]{article}
\usepackage[ruled,lined]{algorithm2e}
\usepackage{enumerate,a4wide}
\usepackage{microtype}
\usepackage{tikz}
\usepackage{amsmath}
\usepackage{amsthm}
\usepackage{amssymb}
\usepackage{hyperref}
\usepackage{authblk,url}
\usepackage[utf8]{inputenc}
\usepackage{color}
\definecolor{myYellow}{rgb}{0.9,0.9,0}

\usetikzlibrary{calc,arrows,decorations.pathreplacing,positioning,backgrounds,calc,trees}
\tikzstyle{knot}=[circle, fill,scale=0.4]
\tikzstyle{grayKnot}=[circle,fill=gray,fill opacity=0.3,text opacity=1,scale=0.9]

\pgfdeclarelayer{background}
\pgfsetlayers{background,main}

\makeatletter
\@ifpackagelater{tikz}{2013/12/01}{%
  \def\AtanTwo(#1,#2){atan2({#1},{#2})}%
}{%
  \def\AtanTwo(#1,#2){atan2({#2},{#1})}%
}
\makeatother

\newcommand{\convexpath}[2]{
  [   
  create hullcoords/.code={
    \global\edef\namelist{#1}
    \foreach [count=\counter] \nodename in \namelist {
      \global\edef\numberofnodes{\counter}
      \coordinate (hullcoord\counter) at (\nodename);
    }
    \coordinate (hullcoord0) at (hullcoord\numberofnodes);
    \pgfmathtruncatemacro\lastnumber{\numberofnodes+1}
    \coordinate (hullcoord\lastnumber) at (hullcoord1);
  },
  create hullcoords
  ]
  ($(hullcoord1)!#2!-90:(hullcoord0)$)
  \foreach [
  evaluate=\currentnode as \previousnode using \currentnode-1,
  evaluate=\currentnode as \nextnode using \currentnode+1
  ] \currentnode in {1,...,\numberofnodes} {
    let \p1 = ($(hullcoord\currentnode) - (hullcoord\previousnode)$),
    \n1 = {\AtanTwo(\y1,\x1) + 90},
    \p2 = ($(hullcoord\nextnode) - (hullcoord\currentnode)$),
    \n2 = {\AtanTwo(\y2,\x2) + 90},
    \n{delta} = {Mod(\n2-\n1,360) - 360}
    in 
    {arc [start angle=\n1, delta angle=\n{delta}, radius=#2]}
    -- ($(hullcoord\nextnode)!#2!-90:(hullcoord\currentnode)$) 
  }
}

\title{Constructing small tree grammars and small circuits for formulas\thanks{The fourth and fifth author are supported
by the DFG research project QUANT-KOMP (Lo 748/10-1).}}

\author[1]{Moses Ganardi\thanks{ganardi@eti.uni-siegen.de}}
\author[1]{Danny Hucke\thanks{hucke@eti.uni-siegen.de}}
\author[2]{Artur Je\.{z}\thanks{aje@cs.uni.wroc.pl}}
\author[1]{Markus Lohrey\thanks{lohrey@eti.uni-siegen.de}}
\author[1]{Eric Noeth\thanks{eric.noeth@eti.uni-siegen.de}}
\affil[1]{University of Siegen, Germany}
\affil[2]{University of Wroc\l{}aw, Poland}

\date{\today}

\newcommand{\mc}{\mathcal}

\newcommand{\rank}{\mathrm{rank}}

\newcommand{\val}{\mathrm{val}}

\newcommand{\labels}{\mathrm{labels}}

\newcommand{\valid}{\mathrm{valid}}

\newcounter{dummy}
\newtheorem{proposition}[dummy]{Proposition}
\newtheorem{example}[dummy]{Example}
\newtheorem{lemma}[dummy]{Lemma}
\newtheorem{theorem}[dummy]{Theorem}
\newtheorem{corollary}[dummy]{Corollary}

\begin{document}

\maketitle

\begin{abstract}
It is shown that every tree of size $n$ over a fixed set of $\sigma$ different ranked symbols can be decomposed 
(in linear time as well as in logspace) into $O\big(\frac{n}{\log_\sigma n}\big) = O\big(\frac{n \log \sigma}{\log n}\big)$ many hierarchically defined pieces.
Formally, such a hierarchical decomposition has the form of a straight-line linear context-free
tree grammar of size $O\big(\frac{n}{\log_\sigma n}\big)$, which can be used as a compressed representation of 
the input tree. This generalizes an analogous result for strings.  Previous grammar-based tree compressors
were not analyzed for the worst-case size of the computed grammar, except for the top dag of Bille et al. \cite{BilleGLW13},
for which only the weaker upper bound of $O\big(\frac{n}{\log_\sigma^{0.19} n}\big)$ (which was very recently improved to 
$O\big(\frac{n \cdot \log \log_\sigma n}{\log_\sigma n}\big)$ \cite{HSR15})
for unranked and unlabelled trees has been derived.
The main result is used to show that every arithmetical formula of size $n$, in which only $m \leq n$ different
variables occur, can be transformed (in linear time as well as in logspace) 
into an arithmetical circuit of size $O\big(\frac{n \cdot \log m}{\log n}\big)$ and depth $O(\log n)$. This refines a classical result of
Brent from 1974, according to which an arithmetical formula of size $n$ can be transformed into a logarithmic depth circuit of size $O(n)$.
A short version of this paper appeared as \cite{HuLoNo14}.
\end{abstract}

\section{Introduction}

\noindent
{\bf Grammar-based string compression.}
{\em Grammar-based compression} has emerged to an active field in string compression during the past 20 years.
The idea is to represent a given string $s$ by a small context-free grammar that generates only $s$;
such a grammar is also called a {\em straight-line program}, briefly SLP. For instance, the word $(ab)^{1024}$ can be represented
by the SLP with the rules $A_0 \to ab$ and $A_i \to A_{i-1} A_{i-1}$ for $1 \leq i \leq 10$ ($A_{10}$ is the start 
symbol). The size of this grammar is much smaller than the size (length) of the string $(ab)^{1024}$. In general,
an SLP of size $n$ (the size of an SLP is usually defined as the total length of all right-hand sides of the rules)
can produce a string of length $2^{\Omega(n)}$. Hence, an SLP can be seen indeed as a succinct representation
of the generated string. The goal of grammar-based string compression is to construct from a given input
string $s$ a small SLP that produces $s$. Several algorithms for this have been proposed and analyzed. 
Prominent grammar-based string compressors are for instance 
{\sf LZ78}, {\sf RePair}, and {\sf BiSection}, see \cite{CLLLPPSS05} for more details.

To evaluate the compression performance of a grammar-based compressor $\mathcal{C}$, two different approaches can be found
in the literature:
\begin{enumerate}[(a)]
\item One can analyze the maximal size of SLPs produced by $\mathcal{C}$
on strings of length $n$ over the alphabet $\Sigma$ (the size of  $\Sigma$ is considered to be a constant larger than one in the further discussion). 
Formally, let
$$
\sigma_{\mathcal{C}}(n) = \max_{x \in \Sigma^n} |\mathcal{C}(x)|,
$$
where $\mathcal{C}(x)$ is the SLP produced by $\mathcal{C}$ on input $x$, and $|\mathcal{C}(x)|$ is the size 
of this SLP. 
An information-theoretic argument shows that for almost all binary strings of length $n$ (up to an exponentially small part) the smallest
SLP has size $\Omega\big(\frac{n}{\log n}\big)$.\footnote{If we do not specify the base of a logarithm, we always mean $\log_2(n)$.}
Explicit examples of strings for which the smallest SLP has size $\Omega\big(\frac{n}{\log n}\big)$ 
result from de Bruijn sequences; see Section~\ref{sec-strings}.
 On the other hand, for many grammar-based compressors $\mathcal{C}$ one gets
$\sigma_{\mathcal{C}}(n) \in O\big(\frac{n}{\log n}\big)$ and hence
$\sigma_{\mathcal{C}}(n) \in \Theta\big(\frac{n}{\log n}\big)$. This holds for instance for the above mentioned
{\sf LZ78}, {\sf RePair}, and {\sf BiSection}, and in fact for all compressors that produce so-called irreducible SLPs \cite{KiYa00}. This 
fact is used in \cite{KiYa00} to construct universal string compressors based on grammar-based compressors.
\item A second approach is to analyze the size of the SLP produced by $\mathcal{C}$ for an input string $x$ compared
to the size of a smallest SLP for $x$. This leads to the approximation ratio for $\mathcal{C}$, which is formally 
defined as
$$
\alpha_{\mathcal{C}}(n) = \max_{x \in \Sigma^n} \frac{|\mathcal{C}(x)|}{g(x)},
$$
where $g(x)$ is the size of a smallest SLP for $x$. It is known that unless $\mathsf{P=NP}$, there is no polynomial time 
grammar-based compressor $\mathcal{C}$ such that $\alpha_{\mathcal{C}}(n) < 8569/8568$ for all $n$ \cite{CLLLPPSS05}.
The  best known  polynomial time 
grammar-based compressors  \cite{CLLLPPSS05,Jez13approx,Jez14approx,Ryt03,Sakamoto05} have
an approximation ratio of  $\mathcal{O}(\log (n/g))$, where $g$ is the size of a smallest
SLP for the input string and each of them works in linear time. 
\end{enumerate}

\smallskip
\noindent
{\bf Grammar-based tree compression.}
In this paper, we want to follow approach (a), but for trees instead of strings.
A tree in this paper is always a rooted ordered tree over a ranked alphabet, i.e., every
node is labelled with a symbol, the rank of this symbol is equal to the number of children of the node
and there is an ordering among the children of a node.
In \cite{BuLoMa07}, grammar-based compression was extended from strings to trees. 
For this, linear context-free tree grammars were used. Linear context-free tree grammars that produce only a single tree
are also known as tree straight-line programs (TSLPs) or  straight-line context-free tree grammars (SLCF tree grammars).
TSLPs generalize dags (directed acyclic graphs), which are widely used as a compact tree representation. Whereas
dags only allow to share repeated subtrees, TSLPs can also share repeated internal tree patterns.

Several grammar-based tree compressors were developed in \cite{Akutsu10,MLMN13,BuLoMa07,JezLo14approx,LohreyMM13},
where the work from \cite{Akutsu10} is based on another type of tree grammars (elementary ordered tree grammars).
The algorithm from \cite{JezLo14approx} achieves an approximation ratio of $O(\log n)$ (for a constant set of node labels).
On the other hand, for none of the above mentioned compressors it is known, whether for every input tree with $n$ nodes 
the size of the output grammar is bounded by $O\big(\frac{n}{\log n}\big)$, as it is the case for many grammar-based string compressors. 
Recently, it was shown that the so-called  {\em top dag} of an unranked and unlabelled tree of size $n$ has size $O\big(\frac{n}{\log^{0.19} n}\big)$
\cite{BilleGLW13} and this bound has been subsequently improved to 
$O\big(\frac{n \cdot \log \log_\sigma n}{\log_\sigma n}\big)$ in \cite{HSR15}.
The top dag can be seen as a slight variant of a TSLP for an unranked tree.

In this paper, we present a grammar-based tree compressor that transforms a given node-labelled tree
of size $n$ with $\sigma$ different node labels, whose rank is bounded by a constant,
into a TSLP of size $O\big(\frac{n}{\log_\sigma n}\big)$ and depth $O(\log n)$, where the depth of a TSLP is the depth
of the corresponding derivation tree (we always assume that $\sigma \geq 2$). In particular, for an unlabelled binary tree we get a 
TSLP of size $O\big(\frac{n}{\log n}\big)$.
Our compressor is basically an extension of the {\sf BiSection} algorithm \cite{KiefferYNC00} 
from strings to trees and is therefore called {\sf TreeBiSection}. It 
works in two steps:\footnote{The following outline works only for binary trees, but it 
can be easily adapted to trees of higher rank, as long as the rank is bounded by a constant.}

In the first step, {\sf TreeBiSection} hierarchically decompose in a top-down way the input tree into pieces of roughly equal size.
This is a well-known technique that is also known as the $(1/3, 2/3)$-Lemma \cite{LeStHa65}. But care has to be taken
to bound the ranks of the nonterminals of the resulting TSLP. As soon as we get a tree with three holes
during the decomposition (which corresponds in the TSLP to a nonterminal of rank three) we have to do 
an intermediate step that decomposes the tree into two pieces having only two holes each. This may involve
an unbalanced decomposition. On the other hand, such an unbalanced decomposition is only necessary in 
every second step. This trick to bound the number of holes by three was used by Ruzzo \cite{Ru80}  in his analysis
of space-bounded alternating Turing machines.

The TSLP produced in the first step can be identified with its derivation tree, which has logarithmic depth. Thanks to the fact 
that all nonterminals have rank at most three, we can encode the derivation tree by a tree with $O(\sigma)$ many
labels. Moreover, this derivation tree is weakly balanced in the following sense. For each edge $(u,v)$ 
in the derivation tree such that both $u$ and $v$ are internal nodes, the derivation tree is balanced at $u$ or $v$.
In a second step, {\sf TreeBiSection} computes the minimal dag of the derivation tree. Due to its balanced shape,
we can show that this minimal dag has size at most $O\big(\frac{n}{\log_\sigma n}\big)$. The 
nodes of this dag are the nonterminals of our final TSLP.

We prove that the algorithm sketched above can be implemented so that it works in logarithmic space, see Section~\ref{sec-time-space}.
Concerning the running time, we show the upper bound of $O(n \log n)$. 
An alternative algorithm {\sf BU-Shrink} (for bottom-up shrink) with a linear running time is presented in Section~\ref{sec:lin_time}. 
In a first step, it merges nodes of the input tree in a bottom-up way. Thereby it constructs a
partition of the input tree into $O\big(\frac{n}{\log_\sigma n}\big)$ many connected parts of size at most
$c \cdot \log_\sigma n$, where $c$ is a suitably chosen constant.  Each such connected part has 
at most three neighbors, one at the top and two at the bottom. By associating with each such connected
part a nonterminal, we obtain a TSLP for the input tree 
consisting of a start rule $S \to s$, where $s$ consists of $O\big(\frac{n}{\log_\sigma n}\big)$ many 
nonterminals of rank at most two, each having a right-hand side consisting of $c \cdot \log_\sigma n$
many terminal symbols.
By choosing the constant $c$ suitably, we can  
(using the formula for the number of binary trees of size $m$, which is given by the Catalan numbers) bound the number 
of different subtrees of these right-hand sides
by $\sqrt{n} \in O\big(\frac{n}{\log_\sigma n}\big)$. This allows to build up the right-hand sides for the non-start nonterminals 
using $O\big(\frac{n}{\log_\sigma n}\big)$ many nonterminals. A combination of this algorithm with {\sf TreeBiSection} --
we call the resulting algorithm {\sf BU-Shrink+TreeBiSection} --
finally yields a linear time algorithm for constructing a TSLP of size 
$O\big(\frac{n}{\log_\sigma n}\big)$ and depth $O(\log n)$.

Let us remark that our size bound $O\big(\frac{n}{\log_\sigma n}\big)$ does not contradict any information-theoretic lower bounds (it actually
matches the information theoretic limit):
Consider for instance unlabelled ordered trees. There are roughly $4^n/\sqrt{\pi n^3}$ such trees with $n$ nodes. Hence under any binary
encoding of unlabelled trees, most trees are mapped to bit strings of length at least $2n - o(n)$.
But when encoding a TSLP of size $m$ into a bit string, another $\log(m)$-factor arises. Hence, a 
TSLP of size $O\big(\frac{n}{\log n}\big)$ is encoded by a bit string of size $O(n)$.

It is also important to note that our size bound $O\big(\frac{n}{\log_\sigma n}\big)$ only holds for trees, where the maximal rank
is bounded by a constant. In particular, it does not directly 
apply to unranked trees (that are, for instance, the standard tree model for XML), which is in contrast to top dags.
To overcome this limitation, one can transform an unranked tree of size $n$ into its
first-child-next-sibling encoding \cite[Paragraph 2.3.2]{Knuth68}, 
which is a ranked tree of size $n$. Then, the first-child-next-sibling encoding
can be transformed into a TSLP of size $O\big(\frac{n}{\log_\sigma n}\big)$.

\medskip
\noindent
{\bf Transforming formulas into circuits.}
Our main result has an interesting application for the classical problem of transforming formulas into small circuits, which will be presented
in Section~\ref{sec-circuits}.
Spira \cite{Spira71} has shown that for every Boolean formula of size $n$ there exists an equivalent Boolean circuit of depth
$O(\log n)$ and size $O(n)$. Brent \cite{Brent74} extended Spira's theorem to formulas over arbitrary semirings and moreover
improved the constant in the $O(\log n)$ bound. Subsequent improvements that mainly concern constant factors can be found in \cite{BonetB94,BshoutyCE95}.
An easy corollary of our $O\big(\frac{n}{\log_\sigma n}\big)$ bound for TSLPs is that for every 
(not necessarily commutative) semiring (or field), every formula of size $n$, in which only $m \leq n$ different
variables occur, can 
be transformed into a circuit of depth $O(\log n)$ and size $O\big(\frac{n \cdot \log m}{\log n}\big)$. Hence, we refine the size bound
from $O(n)$ to $O\big(\frac{n \cdot \log m}{\log n}\big)$ (Theorem~\ref{thm-circuit}). The transformation can be achieved in logspace and,
alternatively, in linear time.
Another interesting point of our formula-to-circuit conversion is that 
most of the construction (namely the construction of a TSLP for the input formula) is purely syntactic.
The remaining part (the transformation of the TSLP into a circuit) is straightforward. In contrast, the
constructions from \cite{BonetB94,Brent74,BshoutyCE95,Spira71} construct a log-depth circuit from a formula
in one step.

\medskip
\noindent
{\bf Related work.}
Several papers deal with algorithmic problems on trees that are succinctly represented by TSLPs,
see \cite{Loh12survey} for a survey. Among other problems, equality checking and the evaluation of tree
automata can be done in polynomial time for TSLPs. 

It is interesting to compare our $O\big(\frac{n}{\log_\sigma n}\big)$ bound 
with the known bounds for dag compression. A 
counting argument shows that for all unlabelled binary trees of size $n$, 
except for an exponentially small part,
the size of a smallest TSLP is $\Omega\big(\frac{n}{\log n}\big)$, and hence
(by our main result) $\Theta\big(\frac{n}{\log n}\big)$. This implies that also the average size of the minimal TSLP, where the average
is taken for the uniform distribution on unlabelled binary trees of size $n$, is $\Theta\big(\frac{n}{\log n}\big)$. In contrast, the average size of 
the minimal dag for trees of size $n$ is $\Theta\big(\frac{n}{\sqrt{\log n}}\big)$ \cite{FlajoletSS90}, whereas the worst-case size of the dag is
 $n$.
 
 In the context of tree compression, succinct data structures for trees are another big topic. There, the goal
is to represent a tree in a number of bits that asymptotically matches the information
theoretic lower bound, and at the same time allows efficient querying (in the best case in time $\mc O(1)$) of the data structure.
For unlabelled unranked trees of size $n$ there exist representations with $2n+o(n)$ bits that support navigation and some other tree queries
in time $\mc O(1)$ \cite{Jacobson89,MunroR01}. This result has been extended to labelled trees, where $(\log \sigma) \cdot n + 2n + o(n)$ 
bits suffice when $\sigma$ is the number of node labels \cite{FerraginaLMM09}. See \cite{RaRa13} for a survey.
 
In view of Theorem~\ref{thm-circuit} on arithmetical circuits, we should mention the following
interesting difference between commutative and noncommutative semirings:
By a classical result of Valiant et al., every  arithmetical circuit of polynomial size and degree over a commutative semiring
can be restructured into an equivalent unbounded fan-in arithmetical circuit of polynomial size and logarithmic depth \cite{ValiantSBR83}.
This result fails in the noncommutative case:
In \cite{Kosaraju90},  Kosaraju gave an example of a circuit family of linear size and degree over a 
noncommutative semiring that is not equivalent to a polynomial size circuit family of depth $o(n)$.
In Kosaraju's example, addition (multiplication, respectively) is the union (concatenation, respectively) of languages. A similar
example was given by Nisan in \cite{Nisan91}.

\section{Computational models}

We will consider time and space bounds for computational problems. For time bounds, we will use the standard
RAM model. We make the assumption that for an input tree of size $n$, arithmetical operations on numbers with $O(\log n)$ bits can be carried
out in time $O(1)$.  We assume that the reader has some familiarity with logspace computations, see e.g. \cite[Chapter~4.1]{AroraBarak} for more details.
A function can be computed in logspace, if it can be computed on a Turing machine with three tapes: a read-only input tape, a write-only output tape, and 
a read-write working tape of length $O(\log n)$, where $n$ is the length of the input. It is an important fact that if functions $f$ and $g$ can be computed
in logspace, then the composition of $f$ and $g$ can be computed in logspace as well. We will use this fact implicitly all over the paper.

\section{Strings and Straight-Line Programs}
\label{sec-strings}

Before we come to grammar-based tree compression, let us briefly discuss grammar-based
string compression. A \emph{straight-line program}, briefly SLP, is a context-free grammar that
produces a single string. Formally, it is defined as a tuple $\mc G = (N,\Sigma, S, P)$, where $N$ is a 
finite set of nonterminals, $\Sigma$ is a finite set of terminal symbols ($\Sigma \cap N = \emptyset$), 
$S \in N$ is the start nonterminal, and $P$ is a finite
set of rules of the form $A \to w$ for $A \in N$, $w \in (N \cup \Sigma)^*$
such that the following conditions hold:
\begin{itemize}
\item There do not exist rules $(A \to u)$ and $(A \to v)$ in $P$ with $u \neq v$.
\item The binary relation $\{ (A, B) \in N \times N \mid (A \to w) \in P,\;B \text{ occurs in } w \}$
is acyclic.
\end{itemize}
These conditions ensure that every nonterminal $A \in N$ produces a unique string $\val_{\mc G}(A) \in \Sigma^*$.
The string defined by $\mc G$ is $\val(\mc G) = \val_{\mc G}(S)$. The size of the SLP $\mc G$ is 
$|\mc G| = \sum_{(A \to w) \in P} |w|$, where $|w|$ denotes the length of the string $w$. 
SLPs are also known as {\em word chains} in the area of combinatorics on words \cite{BeBr87,Diw86}.

A simple induction shows that for every SLP $\mc G$ of size $m$ one has
$|\val(\mc G)| \leq \mathcal{O}(3^{m/3})$ \cite[proof of Lemma~1]{CLLLPPSS05}. On the other hand, it is straightforward to define an SLP
$\mc H$ of size $2n$ such that $|\val(\mc H)| \geq 2^n$. 
This justifies to see an SLP $\mc G$ as a compressed representation of the string $\val(\mc G)$, 
and exponential compression rates can be achieved in this way.

Let $\sigma \geq 2$ be the size of the terminal alphabet $\Sigma$.\footnote{The case $\sigma=1$ is not interesting, since
every string of length $n$ over a unary alphabet can be produced by an SLP of size $O(\log n)$.}
It is well-known that for every string $x \in \Sigma^*$ of length $n$ 
there exists an SLP $\mc G$ of size $O\big(\frac{n}{\log_\sigma n}\big)$ such that $\val(\mc G) = x$, see e.g. \cite{KiYa00}.
On the other hand, an information-theoretic argument shows that for all
strings of length $n$, except for an exponentially small part, the smallest SLP has size $\Omega\big(\frac{n}{\log_\sigma n}\big)$. 
For SLPs, one can, in contrast to other models like Boolean circuits, construct explicit strings that achieve
this worst-case bound as the following result shows. For $\sigma \geq 3$ this result is also shown in \cite{BeBr87}, using a slightly
different argument. 

\begin{proposition}
Let $\Sigma$ be an alphabet of size $\sigma \geq 2$. For every $n\ge \sigma^2$, one can construct in time polynomial in $n$ and $\sigma$
a string $s_{\sigma,n}\in\Sigma^*$ of length $n$ such that every SLP for $s_{\sigma,n}$ has size $\Omega\big(\frac{n}{\log_\sigma n}\big)$.
\end{proposition}

\begin{proof}
Let $r = \lceil\log_\sigma n\rceil \geq 2$. The sequence $s_{\sigma,n}$ is in fact a prefix of a de Bruijn sequence \cite{deBr46}. 
Let $x_1, \ldots, x_{\sigma^{r-1}}$ be a list of all words from $\Sigma^{r-1}$. 
Construct a directed graph by taking these strings as vertices and drawing an $a$-labelled edge 
($a \in \Sigma$) from $x_i$ to $x_j$ if $x_i = b w$ and $x_j = wa$ for some $w \in \Sigma^{r-2}$ and 
$b \in \Sigma$. This graph has $\sigma^r$ edges and 
every vertex of this graph has indegree and outdegree $\sigma$. Hence, it has
a Eulerian cycle, which can be viewed as a sequence $u, b_1, b_2, \ldots, b_{\sigma^r}$, where $u \in \Sigma^{r-1}$ is the start vertex,
and the edge traversed in the $i^\mathrm{th}$ step is labelled with $b_i \in \Sigma$. Define $s_{\sigma,n}$ 
as the prefix of $u b_1 b_2 \cdots b_{\sigma^r}$ of length $n$. The construction implies that
$s_{\sigma,n}$ has $n-r+1$ different substrings of length $r$.
 By the so-called $mk$-Lemma from \cite{CLLLPPSS05}, every SLP for $s_{\sigma,n}$ has size at least
 $$
 \frac{n-r+1}{r} > \frac{n}{r} - 1 \geq \frac{n}{\log_\sigma(n)+1} - 1.
 $$
 This proves the proposition.
 \end{proof}
In  \cite{ArvindRS14} a set of $n$ binary strings of length $n$ is constructed such that any concatenation
circuit that computes this set has size  $\Omega\big(\frac{n^2}{\log^2 n}\big)$. A concatenation circuit for a set $S$ of strings is simply an SLP such that
every string from $S$ is derived from a nonterminal of the SLP. Using the above construction, this lower bound can
be improved to $\Omega\big(\frac{n^2}{\log n}\big)$: Simply take the string $s_{2,n^2}$ and write it as $s_1 s_2 \cdots s_n$ with 
$|s_i|=n$. Then any concatenation circuit for $\{s_1, \ldots, s_n\}$ has size $\Omega\big(\frac{n^2}{\log n}\big)$.

\section{Trees and Tree Straight-Line Programs} \label{sec-trees}

For every $i \ge 0$, we fix a countably infinite set $\mc F_i$ of \emph{terminals} of rank $i$ 
and a countably infinite set $\mc N_i$ of \emph{nonterminals} of rank $i$. Let
$\mc F=\bigcup_{i\ge 0} \mc F_i$ and $\mc N=\bigcup_{i\ge 0} \mc N_i$.
Moreover, let $\mc X=\left\{x_1,x_2,\dots\right\}$ be a countably infinite set of \emph{parameters}.
We assume that the three sets $\mc F$, $\mc N$, and $\mc X$ are pairwise disjoint.
A \emph{labelled tree} $t = (V,\lambda)$ is a finite, rooted and ordered tree $t$ with node set $V$, whose nodes are labelled by elements from $\mc F\cup\mc N\cup \mc X$.
The function $\lambda : V \to \mc F\cup\mc N\cup \mc X$ denotes the labelling function.
We require that a node $v \in V$ with $\lambda(v) \in \mc F_k\cup \mc N_k$ has exactly $k$ children,
which are ordered from left to right. 
We also require that every node $v$ with $\lambda(v) \in \mc X$ is a leaf of $t$.
The size of $t$ is $|t| = |\{ v \in V \mid \lambda(v) \in \mc F\cup\mc N\}|$, i.e., we do not
count parameters.

We denote trees in their usual term notation, e.g.\ $b(a,a)$
denotes the tree with root node labelled by $b$ and two children, both labelled by $a$.
We define $\mc T$ as the set of all labelled trees.
The \emph{depth} of a tree $t$ is the maximal length (number of edges) of a path from the root to a leaf. 
Let $\labels(t) = \{ \lambda(v) \mid  v \in V\}$. For $\mc L \subseteq \mc F\cup\mc N\cup \mc X$ we let
 $\mc T(\mc L) = \{ t \mid \labels(t) \subseteq \mc L\}$.
We write $<_t$ for the depth-first-order on $V$. Formally, $u <_t v$ if $u$ is an ancestor of $v$ or if there exists a node
 $w$ and $i<j$ such that the $i^\mathrm{th}$ child of $w$
 is an ancestor of $u$ and the $j^\mathrm{th}$ child of $w$ is an ancestor of $v$.
 The tree $t\in \mc T$ is a \emph{pattern} if there do not exist different nodes that are labelled with the same parameter.
 For example, $f(x_1,x_1,x_3)$ is not a pattern, whereas $f(x_1,x_{21},x_{99})$ and  $f(x_1,x_2,x_3)$ are patterns.
A pattern $t\in \mc T$ is \emph{valid} if (i) $\labels(t)\cap \mc X=\left\{x_1,\dots,x_n\right\}$ for some $n\ge 0$
 and (ii) for all $u,v \in V$ with $\lambda(u) = x_i$, $\lambda(v) = x_j$ and $u <_t v$ we have $i < j$. 
 For example the pattern $f(x_1,x_{21},x_{99})$ is not valid, whereas $f(x_1,x_2,x_3)$ is valid.
 For a pattern $t$ we define $\valid(t)$ as the unique valid pattern which is obtained from $t$ by renaming the parameters.
 For instance, $\valid(f(x_{21},x_{2},x_{99})) = f(x_1,x_2,x_3)$.
 A valid pattern $t$ in which the parameters $x_1,\dots,x_n$ occur is also written as $t(x_1,\dots,x_n)$ and we write $\rank(t)=n$
 (the {\em rank} of the pattern).
 We say that a valid pattern $p$ of rank $n$ \emph{occurs} in a tree $t$ if there exist $n$ trees
 $t_1,\dots, t_n$ such that $p(t_1,\dots,t_n)$ (the tree obtained from $p$ by replacing the parameter $x_i$ by
 $t_i$ for $1 \leq i \leq n$)
 is a subtree of $t$.
 
 The following counting lemma will be needed several times:
 \begin{lemma} \label{lemma-counting}
The number of trees $t \in \mc T(\mc F)$ with $1 \leq |t| \leq n$ and $|\labels(t)| \leq \sigma$ is 
bounded by $\frac{4}{3} (4\sigma)^n$.
\end{lemma}

\begin{proof}
The number of different rooted ordered (but unranked) trees with $k$ nodes is 
$\frac{1}{k+1} \binom{2k}{k} \leq 4^k$ (the $k^{\text{th}}$ Catalan number, see e.g. \cite{Stan15}).
Hence, the number of trees in the lemma can be bounded by
$$
\sigma^n \cdot \sum_{k=1}^n 4^k  \leq \sigma^n \cdot \frac{4^{n+1}-1}{3} \leq \frac{4}{3} (4\sigma)^n.
$$
\end{proof}

We now define a particular form of context-free tree grammars (see \cite{tata97} for more details on context-free
tree grammars) with the property that exactly one tree is derived.
A \emph{tree straight-line program (TSLP)} is a tuple $\mc G = (N,\Sigma,S,P)$, where 
$N \subseteq \mc N$ is a finite set of nonterminals, $\Sigma \subseteq \mc F$ is a finite set of terminals,
$S \in N \cap \mc N_0$ 
is the start nonterminal, and $P$ is a finite set of rules of the form $A(x_1, \ldots, x_n) \to t(x_1, \ldots, x_n)$ (which is also briefly
written as $A \to t$), where $n \geq 0$, 
$A \in N \cap \mc N_n$ and $t(x_1, \ldots, x_n) \in \mc T(N \cup \Sigma \cup \{x_1, \ldots, x_n\})$ is a valid pattern
such that the following conditions hold:
\begin{itemize}
 \item For every $A \in N$ there is exactly one tree $t$ such that $(A \to t) \in P$.
 \item The binary relation $\{ (A,B) \in \mc N \times \mc N \mid  (A \to t) \in P, B \in  \labels(t)\}$ is acyclic.
\end{itemize}
Note that $N$ and $\Sigma$ are implicitly defined by the rules from $P$. Therefore, we can (and always will)
write a TSLP as a pair $\mc G = (S,P)$ consisting of rules and a start nonterminal.

The above conditions ensure that from every nonterminal $A \in N \cap \mc N_n$ exactly one valid pattern
$\val_{\mc G}(A) \in \mc T(\mc F\cup\{x_1, \ldots, x_n\})$ 
is derived by using the rules as rewrite rules in the usual sense. The tree defined by $\mc G$ is
$\val(\mc G) = \val_{\mc G}(S)$.
Instead of giving a formal definition, we show a derivation of $\val(\mc G)$ from $S$ in an example:
\begin{example}\label{example:SLCF}
Let $\mc G = (S,P)$, $S,A,B,C,D,E,F\in\mc N$, $a\in\mc F_0$, $b\in\mc F_2$ and
\begin{align*}
 P = \{&S\to A(B),\;A(x_1)\to C(F,x_1),\;B\to E(F),\;C(x_1,x_2)\to D(E(x_1),x_2),\\
&D(x_1,x_2)\to b(x_1,x_2),\;E(x_1)\to D(F,x_1),\;F\to a\}.
\end{align*}
A possible derivation of $\val(\mc G) = b(b(a,a),b(a,a))$ from $S$ is:
\begin{align*}
S&\to A(B)\to C(F,B)\to D(E(F),B)\to b(E(F),B)\to b(D(F,F),B)\to b(b(F,F),B)\\
&\to b(b(a,F),B) \to b(b(a,a),B) \to b(b(a,a),E(F))\to b(b(a,a),D(F,F))\\
&\to b(b(a,a),b(F,F))\to b(b(a,a),b(a,F))\to b(b(a,a),b(a,a))
\end{align*}
\end{example}
The size $|\mc G|$ of a TSLP $\mc G=(S,P)$ is the total size of all trees on the right-hand sides of $P$:
\begin{align*}
|\mc G|=\sum_{(A\to t)\in P}|t|
\end{align*}
For instance, the TSLP from Example~\ref{example:SLCF} has size 12.

Note that for the size of a TSLP we do not count nodes of right-hand sides that are labelled with a parameter.
To justify this, we use the following internal representation of valid patterns (which is also used in \cite{JezLo14approx}):
For every non-parameter node $v$ of a tree, with children $v_1, \ldots, v_n$ we store in a list all pairs $(i,v_i)$
such that $v_i$ is a non-parameter node. Moreover, we store for every symbol (node label) its rank. This allows
to reconstruct the valid pattern, since we know the positions where parameters have to inserted.

A TSLP is in {\em Chomsky normal form} if for every rule
$A(x_1, \ldots, x_n) \to t(x_1, \ldots, x_n)$  one of the following two cases holds:
\begin{eqnarray}
t(x_1, \ldots, x_n) & = & B(x_1, \ldots, x_{i-1}, C(x_i, \ldots, x_k), x_{k+1}, \ldots, x_n)   \text{ for } B,C \in \mc N
\label{nonterm-rules} \\
t(x_1, \ldots, x_n) & = & f(x_1, \ldots, x_n) \text{ for } f \in \mc F_n \label{term-rules}.
\end{eqnarray}
If the tree $t$ in the corresponding rule $A\to t$ is of type~\eqref{nonterm-rules}, we write $\mathrm{index}(A)=i$.
If otherwise $t$ is of type~\eqref{term-rules}, we write $\mathrm{index}(A)=0$.
One can transform every TSLP
efficiently into an equivalent TSLP in Chomsky normal form with a small size increase \cite{LoMaSS12}.
We mainly consider TSLPs in Chomsky normal form in the following.

We define the rooted, ordered derivation tree $\mc D_{\mc G}$ of a TSLP $\mc G = (S,P)$
in Chomsky normal form as for string grammars: 
The inner nodes of the derivation tree are labelled by nonterminals
and the leaves are labelled by terminal symbols.
Formally, we start with the root node of $\mc D_{\mc G}$ and assign it the label $S$.
For every node in $D_{\mc G}$ labelled by $A$, where the right-hand side $t$ of the rule for $A$ is of 
type~\eqref{nonterm-rules}, we attach a left child labelled by $B$ and a right
child labelled by $C$. If the right-hand side $t$ of the rule for $A$ is of 
type~\eqref{term-rules}, we attach a single child labelled by $f$ to $A$.
Note that these nodes are the leaves of $\mc D_{\mc G}$ and they 
represent the nodes of the tree $\val(\mc G)$. 
We denote by $\mathrm{depth}(\mc G)$ the depth of the derivation tree $\mc D_{\mc G}$.
For instance, the depth of the TSLP from Example~\ref{example:SLCF} is $4$.

A TSLP is monadic if every nonterminal has rank at most one. The following result was shown in \cite{LoMaSS12}:
\begin{theorem} \label{theo-monadic}
From a given TSLP $\mc G$ in Chomsky normal form such that every nonterminal has rank at most $k$ and every terminal symbol
has rank at most $r$, one can compute in time $O(|\mc G| \cdot k \cdot r)$ a monadic TSLP $\mc H$ with the 
following properties:
\begin{itemize}
\item $\val(\mc G) = \val(\mc H)$,
\item $|\mc H| \in O(|\mc G| \cdot r)$,
\item $\mathrm{depth}(\mc H) \in O(\mathrm{depth}(\mc G))$.
\end{itemize}
Moreover, one can assume that every production of $\mc H$ has one of the following four forms:
\begin{itemize}
\item $A \to B(C)$ for $A, C \in \mc N_0$, $B \in \mc N_1$, 
\item $A(x_1) \to B(C(x_1))$ for $A,B,C \in \mc N_1$,
\item $A \to f(A_1, \ldots, A_n)$ for $f \in \mc F_n$, $A,A_1,\ldots,A_n \in \mc N_0$,
\item $A(x_1) \to f(A_1, \ldots, A_{i-1},x_1, A_{i+1}, \ldots, A_n)$ for $f \in \mc F_n$, $A \in \mc N_1$, $A_1,\ldots,A_n \in \mc N_0$.
\end{itemize}
\end{theorem}

A commonly used tree compression scheme is obtained by writing down repeated subtrees only once. 
In that case all occurrences except for the first are replaced by a pointer to the first one.
This leads to a node-labelled \emph{directed acyclic graph} (dag).
It is known that every tree has a unique minimal dag, which is called 
\emph{the dag} of the initial tree.  An example can be found in 
Figure~\ref{fig:deriv-tree}, where the right graph is the dag of the middle tree.
A dag can be seen as a TSLP where every nonterminal has rank zero: The nonterminals are the nodes of the dag.
A node $v$ with label $f$ and $n$ children $v_1, \ldots, v_n$ corresponds to the rule $v \to f(v_1, \ldots, v_n)$.
The root of the dag is the start variable.
Vice versa, it is straightforward to transform a TSLP, where every variable has rank zero, into an equivalent dag.

The dag of a tree $t$ can be constructed in time $O(|t|)$~\cite{DoSeTa80}.
The following lemma shows that the dag of a tree can be also constructed in logspace.

\begin{lemma} \label{lemma-dag-logspace}
 The dag of a given tree can be computed in logspace.
\end{lemma}

\begin{proof}
Assume that the node set of the input tree $t$ is $\{1, \dots, n\}$. We denote by $t[i]$ the subtree of $t$ rooted at node $i$.
Given two nodes $i,j$ of $t$ one can verify in logspace whether the subtrees $t[i]$ and $t[j]$ are isomorphic (we write $t[i] \cong t[j]$ for this)
by performing a preorder traversal over both trees and thereby comparing the two trees symbol by symbol.

The nodes and edges of the dag of $t$ can be enumerated in logspace as follows.
A node $i$ of $t$ is a node of the dag if there is no $j < i$ with $t[i] \cong t[j]$.
By the above remark, this can be checked in logspace.
Let $i$ be a node of the dag and let $j$ be the $k^{\mathrm{th}}$ child of $i$ in $t$. Then $j'$ is the $k^{\mathrm{th}}$ child
of $i$ in the dag where $j'$ is the smallest number such that $t[j'] \cong t[j]$. 
Again by the above remark this  $j'$ can be found in logspace.
\end{proof}

\section{Constructing a small TSLP for a tree}

  Let $t$ be a tree of size $n$ with $\sigma$ many different node labels.
  In this section we present two algorithms that each construct a TSLP for $t$
  of size $O\big(\frac{n}{\log_\sigma n}\big)$ and depth $O(\log n)$, assuming the maximal rank of symbols is bounded by a constant.
  Our first algorithm {\sf TreeBiSection} achieves this while only using logarithmic space, but needs time $O(n \cdot \log n)$.
  The second algorithm first reduces the size of the input tree and then performs 
   {\sf TreeBiSection} on the resulting tree, which yields a linear running time.

 \subsection{TreeBiSection} \label{sec-construction}

  {\sf TreeBiSection}  uses the well-known idea of splitting a tree recursively into smaller parts of roughly equal size,
  see e.g.~\cite{Brent74,Spira71}.   
  For a valid pattern $t = (V,\lambda) \in \mc T(\mc F \cup \mc X)$ and a node $v \in V$ we denote by $t[v]$ the tree $\valid(s)$, where $s$ is
  the subtree rooted at $v$ in $t$.
  We further write $t\setminus v$ for the tree $\valid(r)$, where $r$ is obtained from $t$ by
  replacing the subtree rooted at $v$ by a new parameter. 
  If for instance $t = h(g(x_1, f(x_2,x_3)), x_4)$ and $v$ is the $f$-labelled node, then 
  $t[v] = f(x_1,x_2)$ and $t \setminus v = h(g(x_1,x_2),x_3)$.
  The following lemma is well-known, at least for binary trees; see e.g. \cite{LeStHa65}. 
  
  \begin{lemma}\label{lemma:r-lemma}
 Let $t \in \mc T(\mc F \cup \mc X)$ be a tree with $|t|\ge 2$ such that every node has at most $r$ children (where $r \geq 1$).
  Then there is a node $v$ such that
  \[
   \frac{1}{2(r+2)} \cdot |t|  \; \le \; |t[v]| \; \le \; \frac{r+1}{r+2} \cdot |t| .
  \]
 \end{lemma}
 
  \begin{proof}
    We start a search at the root node, checking at each node $v$ whether 
    $|t[v]| \leq  \frac{d+1}{d+2} \cdot |t|$, where $d$ is the number of children
   of $v$.  If the property does not hold, we continue the search at a child that spawns a largest subtree 
    (using an arbitrary tie-breaking rule). Note that we eventually reach a node such that $|t[v]| \leq  \frac{d+1}{d+2} \cdot |t|$:
    If $|t[v]| = 1$, then $|t[v]|  \leq  \frac{1}{2} |t|$ since $|t| \geq 2$. 
     
    So, let $v$ be the first node with $|t[v]| \leq  \frac{d+1}{d+2} \cdot |t|$, where $d$ is the number of children of $v$.
   We get 
   $$|t[v]| \leq  \frac{d+1}{d+2} \cdot |t| \leq \frac{r+1}{r+2} \cdot |t|.$$ 
   Moreover $v$ cannot be the root node. Let $u$ be its 
   parent node and let $e$ be the number of children of $u$. Since $v$ spans a largest subtree among the children of $u$, we get
   $e \cdot |t[v]| + 1 \geq |t[u]| \geq \frac{e+1}{e+2} \cdot |t|$, i.e., 
   \begin{eqnarray*}
      |t[v]| & \geq & \frac{e+1}{e (e+2)} \cdot |t| - \frac{1}{e}  \\
      & = & \left(\frac{e+1}{e (e+2)} -  \frac{1}{|t| \cdot e} \right)  \cdot |t|  \\
      & \geq  & \left(\frac{e+1}{e (e+2)} -  \frac{1}{2e}\right) \cdot  |t|  \\
      & =  &  \frac{1}{2(e+2)} \cdot |t|  \\
     &  \geq  & \frac{1}{2(r+2)} \cdot |t| .
   \end{eqnarray*}
  \end{proof}
    For the remainder of this section we refer with $\mathrm{split}(t)$ 
    to the unique node in a tree $t$ which is obtained by the procedure from the proof above. 
		Based on Lemma~\ref{lemma:r-lemma}
    we now construct a TSLP  $\mc G_t =(S,P)$ with $\val(\mc G_t)=t$ 
    for a given tree $t$.  It is {\em not} the final TSLP produced by {\sf TreeBiSection}.  
    For our later analysis, it is important to bound the number of parameters in the
    TSLP $\mc G_t$ by a constant. To achieve this, we use an idea from Ruzzo's paper~\cite{Ru80}. 
    
    We will first present the construction and analysis of $\mc G_t$ only for trees, where every node has at most two children, i.e.,
    we consider  trees from $\mc T(\mc F_0  \cup \mc F_1 \cup \mc F_2)$. Let us write $\mc F_{\leq 2}$ for $\mc F_0  \cup \mc F_1 \cup \mc F_2$.
    In Section~\ref{sec-large-degree}, we will consider trees of larger branching degree.
    For the case that $r=2$, Lemma~\ref{lemma:r-lemma} yields for every tree
    $s \in  \mc T(\mc F_{\leq 2} \cup \mc X)$ with $|s| \geq 2$ a node $v = \mathrm{split}(s)$ such that
    \begin{equation} \label{eq-balanced-splitting}
    \frac{1}{8} \cdot |s|  \; \le \; |s[v]| \; \le \; \frac{3}{4} \cdot |s| .
    \end{equation}
     Consider an input tree $t \in \mc T(\mc F_{\leq 2})$ (we assume that $|t| \geq 2$).
    Every nonterminal of $\mc G_t$ will be of rank at most three.
    Our algorithm stores two sets of rules, $P_{\mathrm{temp}}$ and $P_{\mathrm{final}}$. 
    The set $P_{\mathrm{final}}$ will contain the rules of the TSLP $\mc G_t$ and the rules from $P_{\mathrm{temp}}$ 
    ensure that the TSLP $(S,P_{\mathrm{temp}}\cup P_{\mathrm{final}})$ produces $t$ at any given time of the procedure.
    We start with the initial setting $P_{\mathrm{temp}}=\{S\to t\}$ and $P_{\mathrm{final}}=\emptyset$.
    While $P_{\mathrm{temp}}$ is non-empty we proceed for each rule $(A\to s)\in P_{\mathrm{temp}}$ as follows:
    
    Remove the rule from $P_{\mathrm{temp}}$. Let $A \in \mc N_r$. 
    If $r \leq 2$ we determine the node $v=\mathrm{split}(s)$ in $s$.
    Then we split the tree $s$ into the two trees $s[v]$ and $s\setminus v$.
    Let $r_1=\rank(s[v])$, $r_2=\rank(s\setminus v)$ and let $A_1\in \mc N_{r_1}$ and $A_2\in \mc N_{r_2}$ be fresh nonterminals.
    Note that $r=r_1+r_2-1$.
    If the size of $s[v]$ ($s\setminus v$, respectively)
    is larger than $1$ we add the rule $A_1\to s[v]$ ($A_2\to s\setminus v$, respectively)
    to $P_{\mathrm{temp}}$.
    Otherwise we add it to $P_{\mathrm{final}}$ as a final rule.
    Let $k$ be the number of nodes of $s$ that are labelled by a parameter and that are smaller 
    (w.r.t.~$<_{s}$) than $v$.
    To link the nonterminal $A$ to the fresh nonterminals $A_1$ and $A_2$ we add the rule
    \begin{align*}
      A(x_1,\dots,x_r)\to A_1(x_1,\dots,x_k,A_2(x_{k+1},\dots,x_{k+r_2}),x_{k+r_2+1},\dots,x_{r})
  \end{align*}
    to the set of final rules $P_{\mathrm{final}}$.
    
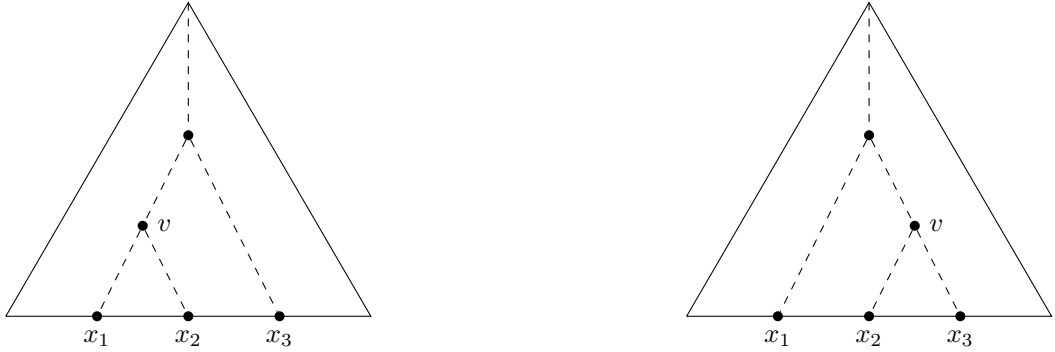
\begin{figure}[t]
 
\centering

\begin{minipage}[hbt]{0.4\textwidth} 
\centering
\begin{tikzpicture}[scale=0.8]

 \coordinate (A) at (0,0);
 \coordinate (B) at (6,0);
 \coordinate (C) at (3,5.2);
 
 \node[knot] (M) at (3,3){};
 
 \node[knot, label=below:$x_1$] (X1) at (1.5,0){}; 
 \node[knot, label=below:$x_2$] (X2) at (3,0){};
 \node[knot, label=below:$x_3$] (X3) at (4.5,0){};

 \node[knot, label=right:$v$] (V) at (2.25,1.5){};
 
 \draw (A)--(B)--(C)--(A);
 \draw[dashed] (C) -- (M) -- (V) -- (X1);
 \draw[dashed] (V) -- (X2);
 \draw[dashed] (M) -- (X3);

\end{tikzpicture}
\end{minipage}
\hfill
\begin{minipage}[hbt]{0.4\textwidth} 
\centering

\begin{tikzpicture}[scale=0.8]

 \coordinate (A) at (0,0);
 \coordinate (B) at (6,0);
 \coordinate (C) at (3,5.2);
 
 \node[knot] (M) at (3,3){};
 
 \node[knot, label=below:$x_1$] (X1) at (1.5,0){}; 
 \node[knot, label=below:$x_2$] (X2) at (3,0){};
 \node[knot, label=below:$x_3$] (X3) at (4.5,0){};

 \node[knot, label=right:$v$] (V) at (3.75,1.5){};
 
 \draw (A)--(B)--(C)--(A);
 \draw[dashed] (C) -- (M) -- (V) -- (X3);
 \draw[dashed] (V) -- (X2);
 \draw[dashed] (M) -- (X1);
 
 \end{tikzpicture}

\end{minipage}
 \caption{Splitting a tree with three parameters}
 \label{fig:threevar}
\end{figure}

    To bound the rank of the introduced nonterminals by three we handle rules $A\to s$ with $A \in \mc N_3$ as follows.
    Let $v_1,v_2$, and $v_3$ be the nodes of $s$ labelled by the parameters $x_1,x_2$, and $x_3$, respectively.
    Instead of choosing the node $v$ by $\mathrm{split}(s)$ we set $v$ to the lowest common ancestor of $(v_1,v_2)$ or $(v_2,v_3)$, 
    depending on which one has the greater distance from the root node (see Figure~\ref{fig:threevar}).
    This step ensures that the two trees $s[v]$ and $s \setminus v$ have rank 2, so in the next step each of these
    two trees will be split in a balanced way according to \eqref{eq-balanced-splitting}. 
     As a consequence, the resulting TSLP $\mc G_t$ has depth $O(\log |t|)$ but size $O(|t|)$.
    
     \begin{example}\label{example:TSLP}
     If we apply our construction to the binary tree $t = b(b(a,a),b(a,a))$ we get the TSLP $\mc G_t = (S,P)$
     with the following rules, where 
     $A,B,C,\ldots, L\in \mc N$, $a\in\mc F_0$, and $b\in\mc F_2$:
\begin{align*}
 P = \{&S\to A(B),\;A(x_1)\to C(D,x_1),\;B\to E(F),\;C(x_1,x_2)\to G(H(x_1),x_2),\;D\to a,\\
&E(x_1)\to I(J,x_1),\;F\to a,\;G(x_1,x_2)\to b(x_1,x_2), \; H(x_1)\to K(L(x_1)),\\
&I(x_1,x_2)\to b(x_1,x_2),\;J\to a,\;K(x_1,x_2)\to b(x_1,x_2),\;L\to a\}.
\end{align*}
\end{example}

   In the next step we want to compact the TSLP by considering the dag of the derivation tree.
   For this we first build the derivation tree $\mc D_t := \mc D_{\mc G_t}$ from the TSLP $\mc G_t$ as described
   above. The  derivation tree for the TSLP $\mc G_t$ from Example~\ref{example:TSLP} is shown on the left of
   Figure~\ref{fig:deriv-tree}.
   
\begin{figure}[t]
\hspace*{\fill}

\begin{minipage}{0.3\textwidth}
\centering
\begin{tikzpicture}[->,level distance=7mm]

\tikzset{level 1/.style={sibling distance=30mm}}
\tikzset{level 2/.style={sibling distance=16mm}}
\tikzset{level 3/.style={sibling distance=11mm}}
\tikzset{level 4/.style={sibling distance=9mm}}

\node  (start){$S$}
  child {node  (a) {$A$}
    child {node  (c) {$C$}
      child {node  (g) {$G$}
        child {node (b1) {$b$}}
      }
      child {node  (h) {$H$}
        child {node  (k) {$K$}    
          child {node (b2) {$b$}}
        }
        child {node  (l) {$L$}
          child {node (a1) {$a$}}
        }
      }
    }
    child {node  (d) {$D$}
      child {node  (f2) {$a$}}
    }
  }
  child {node  (b) {$B$}
    child {node  (e) {$E$}
      child {node  (i) {$I$}
        child {node (b2) {$b$}}
      }
      child {node  (j) {$J$}
        child {node  (f3) {$a$}}
      }
    }
    child {node  (f) {$F$}
      child {node  (f3) {$a$}}
    }
  } 
;
\end{tikzpicture}
\end{minipage}
\hspace*{1.2cm}
\begin{minipage}[hbt]{0.3\textwidth}
\centering
\begin{tikzpicture}[->,level distance=8mm]
\tikzset{level 1/.style={sibling distance=18mm}}
\tikzset{level 2/.style={sibling distance=11mm}}
\tikzset{level 3/.style={sibling distance=9mm}}

\node  (start){$1$}
  child {node  (a) {$1$}
    child {node  (c) {$1$}
      child {node (b1) {$b$}}
      child {node  (h) {$1$}
        child {node (b2) {$b$}}
        child {node (a1) {$a$}}
      }
    }
    child {node  (f2) {$a$}}
  }
  child {node  (b) {$1$}
    child {node  (e) {$1$}
        child {node (b2) {$b$}}
        child {node  (f3) {$a$}}
      }
    child {node (f3) {$a$}}
  }
;

\end{tikzpicture}
\end{minipage}
\hfill
\begin{minipage}[hbt]{0.3\textwidth} 
\centering
\begin{tikzpicture}[->,auto,node distance=11mm,arrow/.style={bend angle=45}]
  \node (1) {1};
  \node (2) [below left of=1] {1};
  \node (3) [below right of=1] {1};
  \node (4) [below left of=2] {1};
  \node (5) [below left of=4] {$b$};
  \node (6) [below right of=4] {1};
  \node (7) [below of=6] {$a$};
  \path
(1) edge (2)
    edge (3)
(2) edge (4)
    edge [bend left=45] (7)
(3) edge (6)
    edge [bend left=45] (7)
(4) edge (5)
    edge (6)
(6) edge (5)
    edge (7);
\end{tikzpicture}
\end{minipage}
\hspace*{\fill}
\caption{The derivation tree from Example~\ref{example:TSLP} }
\label{fig:deriv-tree}
\end{figure}
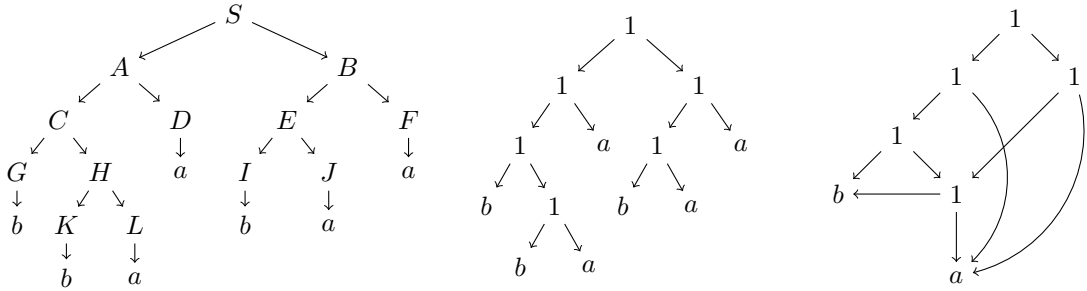

   We now want to identify some (but not all) nonterminals that produce the same tree.
   Note that if we just omit the nonterminal labels from the derivation tree, then there 
   might exist isomorphic subtrees of the derivation whose root nonterminals produce different trees. 
   This is due to the fact that we lost for an $A$-labelled node of the derivation tree with a left (right, respectively) child
   that is labelled with $B$ ($C$, respectively) the information at which argument position of $B$ the nonterminal $C$ is 
   substituted. To keep this information  we replace every label $A$ in the derivation tree with $\mathrm{index}(A) \in \{0,1,2,3\}$
   (the index of a nonterminal of a TSLP in Chomsky normal form was defined in Section~\ref{sec-trees}).
   Moreover, we remove every leaf $v$ and write its label into its parent node.
   We call the resulting tree the \emph{modified derivation tree} and denote it by $\mc D_t^*$.
   Note that $\mc D_t^*$ is a full binary tree with node labels from $\{1,2,3\} \cup \labels(t)$.
   The modified derivation tree for Example~\ref{example:TSLP}
   is shown in the middle of Figure~\ref{fig:deriv-tree}.
   
   The following lemma shows how to compact our grammar by considering the dag of $\mc D_t^*$.

	\begin{lemma}\label{lemma:subtrees-nonterm}
	Let $u$ and $v$ be nodes of $\mc D_t$ labelled by $A$ and $B$, respectively.
	Moreover, let $u'$ and $v'$ be the corresponding nodes in $\mc D_t^*$.
	If the subtrees $\mc D_t^*[u']$ and $\mc D_t^*[v']$ are isomorphic (as labelled ordered trees), then 
	$\val_{\mc G_t}(A)=\val_{\mc G_t}(B)$.
  \end{lemma}
  
  \begin{proof}
	We prove the lemma by induction over the size of the trees $\mc D_t^*[u']$ and $\mc D_t^*[v']$.
	Consider $u$ and $v$ labelled by $A$ and $B$, respectively. We have $\mathrm{index}(A) = \mathrm{index}(B) = i$. 
	For the induction base assume that $i = 0$.
	Then $u'$ and $v'$ are both leaves labelled by the same terminal. Hence, $\val_{\mc G_t}(A)=\val_{\mc G_t}(B)$ holds.
	For the induction step assume that $i > 0$. 	
	Let $A_1$ ($B_1$, respectively) be the label of the left child of $u$ ($v$, respectively) and let
	$A_2$ ($B_2$, respectively) be the label of the right child of $u$ ($v$, respectively). 
	By induction, we get $\val_{\mc G_t}(A_1)=\val_{\mc G_t}(B_1)=s(x_1,\dots,x_m)$ and $\val_{\mc G_t}(A_2)=\val_{\mc G_t}(B_2)=t(x_1,\dots,x_n)$.
	Therefore, $\rank(A)=\rank(B)=m+n-1$ and 
	$\val_{\mc G_t}(A) = s(x_1,...,x_{i-1},t(x_i,...,x_{i+n-1}),x_{i+n},...,x_{n+m-1}) =  \val_{\mc G_t}(B)$.
  \end{proof}
  By Lemma~\ref{lemma:subtrees-nonterm}, if two subtrees of $\mc D_t^*$ are isomorphic we can eliminate the nonterminal of the root node of one subtree.
  Hence, we construct the minimal dag $d$ of $\mc D_t^*$ to compact our TSLP.
  The minimal dag of the TSLP of Example~\ref{example:TSLP} is shown on the right of 
  Figure~\ref{fig:deriv-tree}.
  The nodes of $d$ are the nonterminals of the final TSLP produced by {\sf TreeBiSection}.
  For a nonterminal that corresponds to an inner node of $d$ (a leaf of $d$, respectively), we obtain
  a rule whose right-hand side has the form~\eqref{nonterm-rules} (\eqref{term-rules}, respectively).
  Let $n_1$ be the number of inner nodes of $d$ and $n_2$ be the number of leaves.
  Then the size of our final TSLP is $2n_1+n_2$, which is bounded by twice the number of nodes of $d$.
  The dag from Figure~\ref{fig:deriv-tree} gives the TSLP for the tree $b(b(a,a),b(a,a))$ described in Example~\ref{example:SLCF}.
  Algorithm~\ref{algorithm:SLCF} shows the pseudocode of  {\sf TreeBiSection}. For a valid pattern $s(x_1, \ldots, x_k)$ we denote
  by $\mathrm{lca}_s(x_i,x_j)$ the lowest common ancestor of the unique leafs that are labelled with $x_i$ and $x_j$.
  
  In Section~\ref{sec-time-space} we will analyze the running time of 
  {\sf TreeBiSection}, and we will present a logspace implementation. In Sections~\ref{sec-dag-size} and \ref{sec-size-treebisection}
  we will analyze the size of the produced TSLP.

  \begin{algorithm}[t]
\SetKwComment{Comment}{(}{)}
\SetKwInOut{Input}{input}
\SetKwInOut{Sub}{methods}
\Input{ binary tree $t$}
$P_{\mathrm{temp}} := \{S\to t\}$\\
$P_{\mathrm{final}} := \emptyset$\\
\While{$P_{\mathrm{temp}}\neq\emptyset$}
  {
	\ForEach{$(A\to s)\in P_{\mathrm{temp}}$}
		{
		$P_{\mathrm{temp}} := P_{\mathrm{temp}}\setminus\{A\to s\}$\\
		\eIf{$\rank(s)=3$}
			{ 
		        $v :=$ the lower of nodes  $\mathrm{lca}_s(x_1,x_2), \mathrm{lca}_s(x_2,x_3)$
			}
			{
			$v := \mathrm{split}(s)$
			}
		$t_1 := s[v]$; $t_2 := s\setminus v$ \\
		$r_1 := \rank(t_1)$; $r_2 := \rank(t_2)$\\
		Let $A_1$ and $A_2$ be fresh nonterminals.\\
		\ForEach{$i=1$ \KwTo $2$}
			{
			\eIf{$|t_i|>1$}
				{
				$P_{\mathrm{temp}} := P_{\mathrm{temp}}\cup\{A_i(x_1,\dots,x_{r_i})\to t_i\}$
				}
				{
				$P_{\mathrm{final}} := P_{\mathrm{final}}\cup\{A_i(x_1,\dots,x_{r_i})\to t_i\}$
				}
			}
		$r := r_1+r_2-1$\\
		Let $k$ be the number of nodes in $s$  labelled by parameters that are smaller than $v$ w.r.t. $<_s$.\\
		$P_{\mathrm{final}} := P_{\mathrm{final}}\cup\{A(x_1,\dots,x_r)\to A_1(x_1,\dots,x_k,A_2(x_{k+1},\dots,x_{k+r_2}),x_{k+r_2+1},\dots,x_{r})\}$
		}	
	}
Let $\mc G$ be the TSLP 	$(S,P_{\mathrm{final}})$. \\
Construct the modified derivation tree $\mc D^*_t$ of $\mc G$. \\
Compute  the minimal dag of $\mc D^*_t$ and let $\mc H$ be the corresponding TSLP. \\
  \Return TSLP $\mc H$
\caption{$\mathsf{TreeBiSection}(t,k)$  \label{algorithm:SLCF}}
\end{algorithm}

   \subsection{Running time and space consumption of  TreeBiSection} \label{sec-time-space}
  
  In this section, we show that {\sf TreeBiSection} can be implemented so that it works in logspace, and alternatively
  in time $O(n \cdot \log n)$. Note that these are two different implementations. 
   
        \begin{lemma} \label{lemma-space-time-treebisection}
		Given a tree $t \in \mc T(\mc F_{\leq 2})$ of size $n$ one can compute (i) in time $O(n \log n)$ and (ii)
		in logspace the TSLP produced by {\em TreeBiSection} on input $t$.
	\end{lemma}
	
	\begin{proof}
	         Let $t$ be the input tree of size $n$.
                 The dag of a tree can be computed in (i) linear time \cite{DoSeTa80} and (ii) in logspace by Lemma~\ref{lemma-dag-logspace}.
                 Hence, it suffices to show that the modified derivation tree $\mc D_t^*$ for $t$ can be computed
	        in time $O(n \cdot \log n)$ and in logspace.

	        For the running time let us denote with $P_{\mathrm{temp},i}$ the set of productions $P_{\mathrm{temp}}$ after $i$
	        iterations of the while loop. Moreover, let $n_i$ be the sum of the sizes of all right-hand sides in $P_{\mathrm{temp},i}$.
	        Then, we have $n_{i+1} \leq n_i$:		
	        When a single rule $A \to s$ is replaced with $A_1 \to t_1$ and $A_2 \to t_2$ then each non-parameter node in $t_1$ or $t_2$ is one of the nodes of $s$.
	        Hence, we have $|s| = |t_1|+|t_2|$ (recall that we do not count parameters for the size of a tree).
	        We might have $n_{i+1} < n_i$ since rules with a single terminal symbol on the right-hand side are put into $P_{\mathrm{final}}$.
	        We obtain $n_i \leq n$ for all $i$. Hence, splitting all rules in $P_{\mathrm{temp},i}$ takes time $O(n)$, and a single iteration
	        of the while loop takes time $O(n)$ as well.
		On the other hand, since every second split reduces the size of the tree to which the split is applied by a constant
		factor (see \eqref{eq-balanced-splitting}). Hence, the while loop is iterated at most $O(\log n)$ times.
		This gives the time bound.
		
		The inquisitive reader may wonder whether our convention of neglecting parameter nodes for the size of a tree affects the linear running time.
		This is not the case: Every right-hand side $s$ in $P_{\mathrm{temp},i}$ has at most three parameters, i.e., the total number of nodes in $s$
		is at most $|s|+3 \leq 4|s|$. This implies that the split node can be computed in time $O(|s|)$. Doing this for all right-hand sides in $P_{\mathrm{temp},i}$
		yields the time bound $O(n)$ as above.
			

	        For the logspace version, we
		first describe how to represent a single valid pattern occurring in $t$ in logspace and how to
		compute its split node.
		Let $s(x_1, \dots, x_k)$ be a valid pattern which occurs in $t$ and has $k$ parameters where $0 \le k \le 3$, 
		i.e., $s(t_1, \dots, t_k)$ is a subtree of $t$ for some subtrees $t_1, \dots, t_k$ of $t$.
		We represent the tree $s(x_1, \dots, x_k)$ by the tuple $\mathrm{rep}(s) = (v_0, v_1 \dots, v_k)$
		where $v_0, v_1, \dots, v_k$ are the nodes in $t$ corresponding to the roots of $s, t_1, \dots, t_k$, respectively.
		Note that $\mathrm{rep}(s)$ can be stored using $O((k+1) \cdot \log(n))$ many bits.
		Given such a tuple $\mathrm{rep}(s) = (v_0, \dots, v_k)$, we can compute in logspace
		the size $|s|$ by a preorder traversal of $t$, starting from $v_0$ and skipping subtrees rooted in the
		nodes $v_1, \dots, v_k$. We can also compute in logspace the split node $v$ of $s$:
		If $s$ has at most two parameters, then $v = \mathrm{split}(s)$.
		Note that the procedure from Lemma~\ref{lemma:r-lemma} can be implemented in logspace
		since the size of a subtree of $s$ can be computed as described before.
		If $s$ has three parameters, then $v$ is the lowest common ancestor of either
		$v_1$ and $v_2$, or of $v_2$ and $v_3$, depending on which node has the larger distance from $v_0$.
		The lowest common ancestor of two nodes can also be computed in logspace by traversing the paths
		from the two nodes to the root upwards. From $\mathrm{rep}(s) = (v_0, \dots, v_k)$ and a split node $v$
		we can easily determine $\mathrm{rep}(s[v])$ and $\mathrm{rep}(s \setminus v)$ in logspace.
		
		Using the previous remarks we are ready to present the logspace algorithm to compute $\mc D_t^*$.
		Since $\mc D_t^*$ is a binary tree of depth $O(\log n)$ we can identify a node of $\mc D_t^*$
		with the string $u \in \{0,1\}^*$ of length at most $c \cdot \lfloor \log n \rfloor$ that stores the path from the root to the node, where $c > 0$ is a suitable constant.
		We denote by $s_u$ the tree (with at most three parameters) described by a node $u$ of $\mc D_t^*$ in the sense of Lemma~\ref{lemma:subtrees-nonterm}.
		That is, if $u'$ is the corresponding node of the derivation tree $\mc D_t$ and $u'$ is labelled with the nonterminal
		$A$, then $s_u = \val_{\mc G_t}(A)$.
		
		To compute  $\mc D_t^*$, it suffices 
		for each string $w \in \{0,1\}^*$ of length $c \cdot \lfloor \log n \rfloor$ to check in logspace
		whether it is a node of  $\mc D_t^*$ and in case it is a node, to determine the label of $w$ in  $\mc D_t^*$.
		For this, we compute for each prefix $u$ of $w$, starting with the empty word, the tuple $\mathrm{rep}(s_u)$ and the label of $u$ in $\mc D_t^*$, or a bit 
		indicating that $u$ is not a node of $\mc D_t^*$  (in which case also $w$ is not a node of $\mc D_t^*$).
	         Thereby we only store the current bit strings $w,u$ and the 
		value of $\mathrm{rep}(s_u)$, which fit into logspace.
		If $u = \varepsilon$, then $\mathrm{rep}(s_u)$ consists only of the root of $t$.
		Otherwise, we first compute in logspace the size $|s_u|$ from $\mathrm{rep}(s_u)$.
		If $|s_u| = 1$, then $u$ is a leaf in $\mc D_t^*$ with label $\lambda(u)$ and no longer prefixes
		represent nodes in $\mc D_t^*$.
		If $|s_u| > 1$, then $u$ is an inner node in $\mc D_t^*$ and, as described above, we can compute in logspace from
		$\mathrm{rep}(s_u)$ the tuples $\mathrm{rep}(s_{u0})$ and $\mathrm{rep}(s_{u1})$,
		from which we can easily read off the label of $u$ from $\{1,2,3\}$. If $u=w$, then we stop, otherwise
		we continue with $u i$ and $\mathrm{rep}(s_{ui})$, where $i \in \{0,1\}$ is such that $ui$ is a prefix of $w$.		
	\end{proof} 
	
In view of the above logspace algorithm, it is interesting to remark that Gagie and 
Gawrychowski considered in \cite{GaGa10} the problem of computing in logspace
a small SLP for a given string. They present a logspace algorithm that achieves an
approximation ratio of $O(\min\{g, \sqrt{n/\log n}\})$, where $g$ is the size of a smallest
SLP and $n$ is the length of the input word.

 \subsection{Size of the minimal dag} \label{sec-dag-size}
 
 In order to bound the size  of the TSLP produced by {\sf TreeBiSection}
 we have to bound  the number of nodes in the dag of the modified derivation tree.
 To this end, we prove in this section a 
 general result about the size of dags of certain weakly balanced binary trees that might be of independent interest.
  
  Let $t$ be a binary tree and let $0 < \beta < 1$. 
  The {\em leaf size} of a node $v$ is the number of leaves of the subtree rooted at $v$. 
  We say that an inner node $v$ with children $v_1$ and $v_2$ is 
  {\em $\beta$-balanced} if the following holds: If 
  $n_i$ is the leaf size of $v_i$, then $n_1 \geq \beta n_2$ and $n_2 \geq \beta n_1$.
  We say that $t$ is {\em $\beta$-balanced} if the following holds:
  For all inner nodes $u$ and $v$ such that $v$ is a child of $u$, we have that $u$ is $\beta$-balanced
  or $v$ is $\beta$-balanced.

  \begin{theorem} \label{thm-small-dag}
  If $t$ is a $\beta$-balanced binary tree having $\sigma$ different node
  labels and $n$ leaves $($and hence $|t|, \sigma \leq 2n-1)$, then the size of the dag of $t$ is bounded by 
  $\frac{\alpha \cdot n}{\log_\sigma n}$, where $\alpha \in O(\log_{1+\beta}(\beta^{-1}))$  only depends on $\beta$.
  \end{theorem} 
  
  \begin{proof}
  Let us fix a tree $t = (V,\lambda)$ as in the theorem with $n$ leaves. 
  Moreover, let us fix a number $k$ that will be defined later.
  We first bound the number of different subtrees with at most $k$ leaves in $t$. 
  Afterwards we will estimate the size of the remaining \emph{top tree}. 
  The same strategy is used for instance in \cite{HeapM94,LiawL92}
  to derive a worst-case upper bound on the size of binary decision diagrams.
  
  \medskip
  \noindent
  {\em Claim 1.}  The number of different subtrees of $t$ with at most $k$ leaves
   is bounded by $d^k$ with $d = 16 \sigma^2$ (since $t$ is a binary trees, we could deduce a more precise bound $(4 \sigma^2)^k$, but this is not crucial).
  
   \medskip
  \noindent 
   A subtree of $t$ with at most $k$ leaves has at most $2k-1$ nodes, each of which is labelled with one of 
   $\sigma$ many labels. Hence, by Lemma~\ref{lemma-counting} we can bound the number of subtrees of $t$ with at most $k$ leaves
   by $\frac{4}{3} (4 \sigma)^{2k-1} = \frac{1}{3 \sigma} (4 \sigma)^{2k} \leq (16 \sigma^2)^k$.

   
    \medskip
  \noindent 
  Let $\mbox{top}(t,k)$ be the tree obtained from $t$ by removing all nodes with leaf size at most $k$.
   
  \medskip
  \noindent
  {\em Claim 2.}  The number of nodes of $\mbox{top}(t,k)$ is bounded by 
   $c\cdot\frac{n}{k}$, where $c \in O(\log_{1+\beta}(\beta^{-1}))$ only depends on $\beta$.
   
  \medskip
  \noindent
  The tree $\mbox{top}(t,k)$ has at most $n/k$ leaves since it is obtained from $t$ 
  by removing all nodes with leaf size at most $k$. 
  Each node in $\mbox{top}(t,k)$ has at most two children.
  Therefore, if we can show that the length of  \emph{unary chains} in 
  $\mbox{top}(t,k)$ is bounded by some $c \in O(\log_{1+\beta}(\beta^{-1}))$, then it follows that
  $\mbox{top}(t,k)$ has at most $2 c n/k$ many nodes.

  Let $v_1,\dots, v_m$ be a unary chain in $\mbox{top}(t,k)$ where $v_i$ is the single child node of $v_{i+1}$.
  Moreover, let $v'_i$ be the removed sibling of $v_i$ in $t$, see Figure~\ref{fig:topTree}.
  Note that each node $v'_i$ has leaf size at most $k$. 

  We claim that the leaf size of $v_{2i+1}$ is larger than $(1+\beta)^i k$ for all $i$ with $2i+1 \leq m$.
  For $i=0$ note that $v_1$ has leaf size more than $k$
  since otherwise it would have been removed in $\mbox{top}(t,k)$.
  For the induction step, assume that the leaf size of $v_{2i-1}$ is larger than $(1+\beta)^{i-1} k$.
  One of the nodes $v_{2i}$ and $v_{2i+1}$ must be $\beta$-balanced. 
  Hence, $v'_{2i-1}$ or $v'_{2i}$ must have leaf size more than $\beta(1+\beta)^{i-1} k$.
  Hence, $v_{2i+1}$ has leaf size more than $(1+\beta)^{i-1} k + \beta(1+\beta)^{i-1} k = (1+\beta)^{i} k$.

  Let $\ell = \log_{1+\beta}(\beta^{-1})$. If $m \geq 2\ell+3$, then  $v_{2\ell+1}$ exists and has 
  leaf size more than $(1+\beta)^\ell k = k/\beta$, which implies 
  that the leaf size of $v'_{2\ell+1}$ or $v'_{2\ell+2}$ (both nodes exist) is more than $k$, which is a contradiction.
  Hence, we must have $m \leq 2 \log_{1+\beta}(\beta^{-1}) + 2$, i.e., we can choose
  $c = 2 \log_{1+\beta}(\beta^{-1}) + 2$.
  Figure~\ref{fig:topTree} shows an illustration of our argument.	  
	  
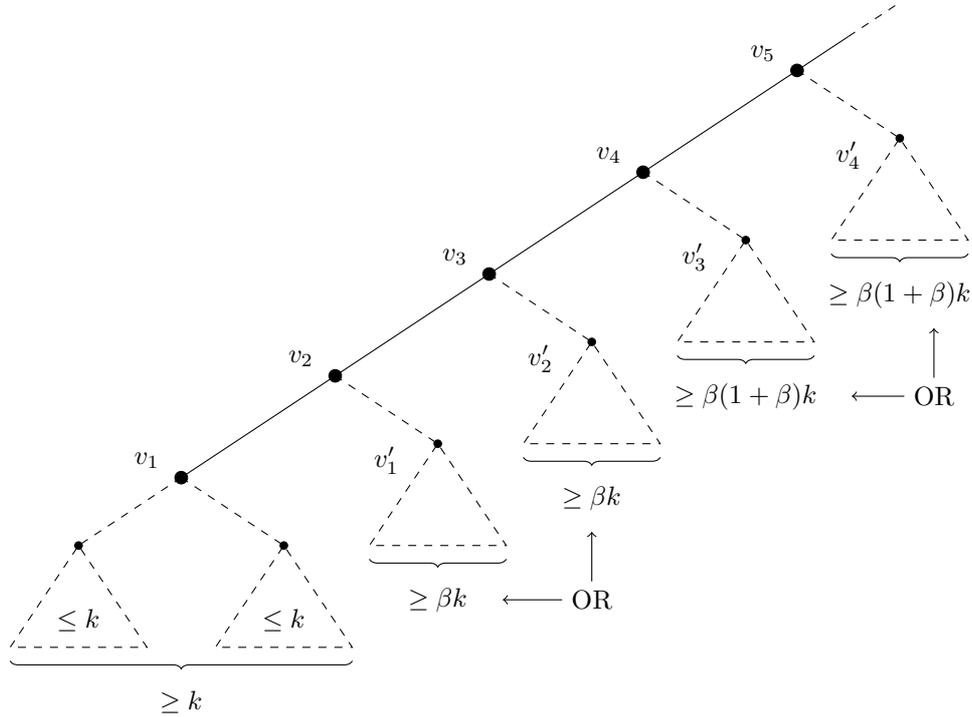
\begin{figure}[t]
\centering
\begin{tikzpicture}[scale=0.9, auto,swap]

 \coordinate (A) at ( 2.5, 2.5);
 \coordinate (B) at (12.25,9);
 \draw (A) -- (B); 
 
 \coordinate (C) at ( 13,9.5);

 \draw [dashed] (B) -- (C);  

 \coordinate (blA) at (0, 0);
 \coordinate (blB) at (2, 0); 
 \node[circle, draw, fill,scale=0.3] (blC) at (1,1.5) {};
 \draw [dashed] (A) -- (blC); 
 \draw [dashed] (blC) -- (blB)  -- (blA) -- (blC);
 \node at (1,0.4) {$\le k$};

 \foreach \x/\y/\geK /\name / \rootlabel in { 
 {2.5/2.5 / $\le k$ / $ $      / $v_1$}, 
 {4.75/4  / $ $     / $v_1'$   / $v_2$},
 {7/5.5   / $ $     / $v_2'$   / $v_3$},
 {9.25/7  / $ $     / $v_3'$   / $v_4$},
 {11.5/8.5/ $ $     / $v_4'$   / $v_5$}}
  {\draw [dashed] (\x,\y) -- (\x+1.5,\y-1) -- (\x +2.5 ,\y-2.5) -- (\x +0.5,\y-2.5) -- (\x+1.5,\y-1); 
   \node at (\x+1.5,\y- 2.1) {\geK};  
   \node at (\x+0.75,\y-1.25) {\name}; 
   \node at (\x-0.5,\y+0.25) {\rootlabel}; 
   \node[circle, draw, fill,scale=0.5] at (\x,\y) {}; 
   \node[circle, draw, fill,scale=0.3] at (\x+1.5,\y-1) {};   
   };
 
 \draw[decorate,decoration={brace,amplitude=3pt,mirror}] 
     (0,-0.2) -- (5,-0.2); 
     
 \foreach \x/\y/\labl in {
  {5.25/1.3/$\ge \beta k$},
  {7.5/2.8/$\ge \beta k$},
  {9.75/4.3/$\ge \beta (1+\beta) k$},
  {12/5.8/$\ge \beta (1+\beta) k$}}
  { \draw[decorate,decoration={brace,amplitude=3pt,mirror}] 
     (\x,\y) -- (\x+2,\y);
     \node at (\x+1,\y-0.6) {\labl }; 
   }; 
  \node at (2.5,-0.8) {$\ge  k$ };
   
  \node (or1) at (8.5,0.7) {OR};
  \node (or2) at (13.5,3.7) {OR};

  \draw[->] (or1) -- (7.2,0.7);
  \draw[->] (or1) -- (8.5,1.7);  
  
  \draw[->] (or2) -- (12.3,3.7);
  \draw[->] (or2) -- (13.5,4.7);  
     
 \end{tikzpicture}
 \caption{A chain within a top tree. The subtree rooted at $v_1$ has more than $k$ leaves.}
 \label{fig:topTree}
\end{figure}

\medskip
\noindent
  Using Claim 1 and 2 we can now prove the theorem:
  The number of nodes of the dag of $t$ is bounded by the number of
  different subtrees with at most $k$ leaves (Claim 1) 
  plus the number of nodes of the remaining tree 
  $\mbox{top}(t,k)$ (Claim 2).
  Let  $k=\frac{1}{2} \log_d n$. Recall that $d = 16 \sigma^2$ and hence $\log d = 4 + 2 \log \sigma$,
  which implies that $\log_d n \in \Theta(\log_\sigma n)$.
  With Claim~1 and 2 we get the following bound on the size of the dag, where $c \in O(\log_{1+\beta}(\beta^{-1}))$ is the bound from Claim~2:
  $$
  d^k + c \cdot \frac{n}{k} = d^{(\log_d n)/2} + 2 c\cdot \frac{n}{\log_d n} 
  = \sqrt{n} + 2c\cdot \frac{n}{\log_d n} \in O \left( \frac{c \cdot n}{\log_d n}\right) 
  = O \left( \frac{c \cdot n}{\log_\sigma n}\right)
  $$
  This proves the theorem.
  \end{proof} 
  Obviously, one could relax the definition of $\beta$-balanced by only requiring
  that if $(v_1, v_2, \ldots, v_\delta)$ is a path down in the tree, where $\delta$ is a constant, then 
  one of the nodes  $v_1, v_2, \ldots, v_\delta$ must be $\beta$-balanced. Theorem~\ref{thm-small-dag} would
  still hold with this definition (with $\alpha$ depending linearly on $\delta$).

Before we apply Theorem~\ref{thm-small-dag} to {\sf TreeBiSection} let us present a few other
results on the size of dags that are of independent interest. If $\beta$ is a constant, then a $\beta$-balanced binary tree $t$ has depth $O(\log |t|)$.
One might think that this logarithmic depth is responsible for the small dag size in Theorem~\ref{thm-small-dag}.
But this intuition is wrong:

  \begin{theorem} \label{thm-large-dag}
  There is a family of trees $t_n \in \mc T(\{a,b,c\})$ with
  $a \in {\mc F}_0$, $b \in {\mc F}_1$, and $c \in {\mc F}_2$
  $(n \geq 1)$ with the following properties:\footnote{The unary node label $b$ can 
  replaced by the pattern $c(d,x)$, where $d \in {\mc F}_0 \setminus \{a\}$ to obtain a binary tree.}
  \begin{itemize}
	\item $|t_n| \in \Theta(n)$
	\item The depth of $t_n$ is  $\Theta(\log n)$. 
	\item The size of the minimal dag of $t_n$ is at least $n$.
  \end{itemize}
  \end{theorem}

\begin{figure}[t]
\centering
\begin{tikzpicture}[scale=1,auto,swap,level distance=3mm]
\tikzset{level 1/.style={sibling distance=32mm}}
\tikzset{level 2/.style={sibling distance=16mm}}
\tikzset{level 7/.style={sibling distance=8mm}}  
\tikzset{level 8/.style={sibling distance=4mm}}
\node (eps) [knot] {} 
  child {node (0) [knot] {} 
	child {node (00) [knot] {}
	  child {node (000) [knot] {}
	    child {node (0 000) [knot] {}
	      child {node (00 000) [knot] {}
	        child {node (000 000) [knot] {}
	          child {node (00000000) [knot] {}
	            child {node (000000000) [knot] {}
	              child {node (0000000000) [knot] {}}
	              child {node (0000000001) [knot] {}}
	            }
	            child {node (000000001) [knot] {}}
	          }
	          child {node (00000001) [knot] {}
	            child {node (000000010) [knot] {}}
	            child {node (000000011) [knot] {}}
	          }
	        }
	      }
	    }  
	  }
	}
    child {node (01) [knot] {}
	  child {node (010) [knot] {}
	    child {node (010 0) [knot] {}
	      child {node (010 00) [knot] {}
	        child {node (010 000) [knot] {}
	          child {node (01000000) [knot] {}
	            child {node (010000000) [knot] {}}
	            child {node (010000001) [knot] {}
	              child {node (0100000000) [knot] {}}
	              child {node (0100000001) [knot] {}}	            
	            }
	          }
	          child {node (01000001) [knot] {}
	            child {node (010000010) [knot] {}}
	            child {node (010000011) [knot] {}}
	          }
	        }
	      }
	    }  
	  }
    }
  }
  child {node (1) [knot] {}
	child {node (10) [knot] {}
	  child {node (100) [knot] {}
	    child {node (1000) [knot] {}
	      child {node (10000) [knot] {}
	        child {node (100 000) [knot] {}
	          child {node (10000000) [knot] {}
	            child {node (100000000) [knot] {}}
	            child {node (100000001) [knot] {}}
	          }
	          child {node (10000001) [knot] {}
	            child {node (100000010) [knot] {}
	              child {node (1000000000) [knot] {}}
	              child {node (1000000001) [knot] {}}	            
	            }
	            child {node (100000011) [knot] {}}
	          }
	        }
	      }
	    }  
	  }
	}
    child {node (11) [knot] {}
      child {node (110) [knot] {}
	    child {node (1100) [knot] {}
	      child {node (11000) [knot] {}
	        child {node (110 000) [knot] {}
	          child {node (11000000) [knot] {}
	            child {node (110000000) [knot] {}}
	            child {node (110000001) [knot] {}}
	          }
	          child {node (11000001) [knot] {}
	            child {node (110000010) [knot] {}}
	            child {node (110000011) [knot] {}
	              child {node (1100000000) [knot] {}}
	              child {node (1100000001) [knot] {}}     
	            }
	          }
	        }
	      }
	    }  
	  }
    }  
  }
;
 \end{tikzpicture}
 \caption{Tree $t_{16}$ from the proof of Theorem~\ref{thm-large-dag}.}
 \label{fig:exTheorem6}
\end{figure}
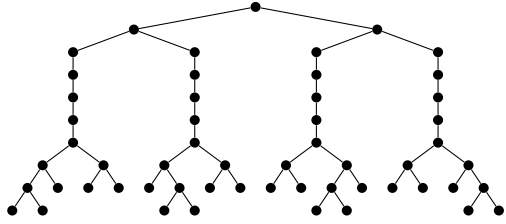

\begin{proof}
  Let $k = \frac{n}{\log n}$ (we ignore rounding problems with $\log n$, which only affect multiplicative factors).
  Choose $k$ different binary trees $s_1, \ldots, s_k \in  \mc T(\{a,c\})$, each having $\log n$ many internal nodes.
  This is possible: For $n$ large enough there are roughly $$\frac{4^{\log n}}{\sqrt{\pi \cdot \log^3 n}} = \frac{n^2}{\sqrt{\pi \cdot \log^3 n}} > n$$ 
  many different binary trees with $\log n$ many internal nodes.
  Then consider the trees $s'_i = b^{\log n}(s_i)$. Each of these trees has size $\Theta(\log n)$ and depth  $\Theta(\log n)$.
  Next, let $u_n(x_1, \ldots, x_k) \in \mc T(\{c,x_1, \ldots, x_k\})$
  be a binary valid pattern (all non-parameter nodes are labelled with $c$) of depth $\Theta(\log k) = \Theta(\log n)$ and size $\Theta(k) = \Theta(\frac{n}{\log n})$.
  We finally take $t_n = u_n(s'_1,\ldots, s'_k)$. Figure~\ref{fig:exTheorem6} shows the tree $t_{16}$.
  We obtain $|t_n| =  \Theta(\frac{n}{\log n}) + \Theta(k \cdot \log n) = \Theta(n)$. The depth of $t_n$ is  $\Theta(\log n)$.
  Finally, in the minimal dag for $t_n$ the unary $b$-labelled nodes cannot be shared. Basically, the pairwise
  different trees $t_1, \ldots, t_n$ work as different constants that are attached to the $b$-chains.
  But the number of $b$-labelled nodes in $t_n$ is $k \cdot \log n = n$.
\end{proof}
Note that the trees from Theorem~\ref{thm-large-dag} are not $\beta$-balanced for any constant $0 < \beta < 1$,
and by Theorem~\ref{thm-small-dag} this is necessarily the case.
Interestingly, if we assume that every subtree $s$ of a binary tree $t$ has depth at most $O(\log |s|)$, 
then H\"ubschle-Schneider and Raman \cite{HSR15} have implicitly shown the bound 
$O\big(\frac{n \cdot \log \log_\sigma n}{\log_\sigma n}\big)$ 
for the size of the minimal dag. 

\begin{theorem}[\cite{HSR15}] \label{thm-top-dag}
Let $\alpha$ be a constant. Then there is a constant $\beta$ that only depends on $\alpha$ such that
the following holds: If $t$ is a binary tree of size $n$ with $\sigma$ many node labels such that every subtree
$s$ of $t$ has depth at most $\alpha \log_2 n + \alpha$, then  the size of the dag of $t$ is at most 
$\frac{\beta \cdot n \cdot \log \log_\sigma n}{\log_\sigma n} + \beta$. 
\end{theorem}

Interestingly, we can show that the bound in this result is sharp:
\begin{theorem} \label{thm-large-dag-2}
  There is a family of trees $t_n \in \mc T(\{a,b,c\})$ with
  $a \in {\mc F}_0$, $b \in {\mc F}_1$, and $c \in {\mc F}_2$
  $(n \geq 1)$ with the following properties:\footnote{Again, the unary node label $b$ can 
  replaced by the pattern $c(d,x)$, where $d \in {\mc F}_0 \setminus \{a\}$ to obtain a binary tree.}
  \begin{itemize}
	\item $|t_n| \in \Theta(n)$
	\item Every subtree $s$ of a tree $t_n$ has depth $O(\log |s|)$.
	\item The size of the minimal dag of $t_n$ is $\Omega(\frac{n \cdot \log \log n}{\log n})$.
  \end{itemize}
  \end{theorem}
  
\begin{proof}
The  tree $t_n$ is similar to the one from the proof of Theorem~\ref{thm-large-dag}. Again, let $k =  \frac{n}{\log n}$.
Fix a balanced binary tree $v_n \in \mc T(\{a,c\})$ with $\log k \in \Theta(\log n)$ many leaves. 
From $v_n$ we construct $k$ many different trees $s_1, \ldots, s_k \in  \mc T(\{a,b,c\})$
by choosing in $v_n$ an arbitrary subset of leaves (there are $k$ such subsets) and replacing all leaves
in that subset by $b(a)$. Note that $|s_i| \in \Theta(\log n)$.  Moreover, every subtree $s$ of a tree $s_i$
has depth $O(\log |s|)$ (since $v_n$ is balanced). 
Then consider the trees $s'_i = b^{\log \log n}(s_i)$ (so, in contrast to the proof of 
Theorem~\ref{thm-large-dag}, the length of the unary chains is $\log \log n$).
Clearly, $|s'_i| \in \Theta(\log n)$. Moreover, we still have the property that
every subtree $s$ of a tree $s'_i$ has depth $O(\log |s|)$: This is clear, if that subtree is rooted in a node
from $s_i$. Otherwise, the subtree $s$ has the form $b^{h}(s_i)$
for some $h \leq \log \log n$. This tree has depth $h + \Theta(\log \log n) = \Theta(\log \log n)$
and size $\Theta(\log n)$. Finally we combine 
the trees $s'_1, \ldots, s'_k \in  \mc T(\{a,b,c\})$ in a balanced way to a single tree using the 
binary valid pattern $u_n(x_1, \ldots, x_k)$ from the proof of Theorem~\ref{thm-large-dag}, i.e.,
$t_n = u_n(s'_1, \ldots, s'_k)$. Then, $|t_n| = \Theta(n)$. Moreover, by the same argument as 
in the proof of Theorem~\ref{thm-large-dag}, the dag for $t_n$ has size $\Omega(\frac{n \cdot \log \log n}{\log n})$
since the nodes in the $k = \frac{n}{\log n}$ many unary chains of length $\log \log n$ cannot be shared with other 
nodes. It remains to show that every subtree $s$ of $t_n$ has depth  $O(\log |s|)$. For the case that $s$ is a subtree of one of the 
trees $s'_i$, this has been already shown above. But the case that $s$ is rooted in a node from  $u_n(x_1, \ldots, x_k)$ is also 
clear: In that case, $s$ is of the form $s = u'(s'_i, \ldots, s'_j)$, where $u'(x_i, \ldots, x_j)$ is  a subtree of $u_n(x_1, \ldots, x_k)$.
Assume that $d$ is the depth of $u'(x_i, \ldots, x_j)$. Since $u_n(x_1, \ldots, x_k)$ is a balanced binary tree,
we have $|u'(x_i, \ldots, x_j)| \in \Omega(2^d)$ and $j-i+1 \in \Omega(2^d)$. Hence, 
$s = u'(s'_i, \ldots, s'_j)$ has size $\Omega(2^d + 2^d \cdot \log n) = \Omega(2^d \cdot \log n)$ 
and depth $d + \Theta(\log \log n)$. This shows the desired property since $\log(2^d \cdot \log n) = d + \log \log n$.
\end{proof}

H\"ubschle-Schneider and Raman \cite{HSR15} used Theorem~\ref{thm-top-dag} to prove the upper bound 
$O\big(\frac{n \cdot \log \log_\sigma n}{\log_\sigma n}\big)$ 
for the size of the top dag of an unranked tree of size $n$ with $\sigma$ many node labels. It is not 
clear whether this bound can be improved to  $O\big(\frac{n}{\log_\sigma n}\big)$. The trees used
in Theorem~\ref{thm-large-dag-2} do not seem to arise as the top trees of unranked trees.

\subsection{Size of the TSLP produced by TreeBiSection} \label{sec-size-treebisection}

  Let us fix the TSLP $\mc G_t$ for a tree $t \in \mc T(\mc F_{\leq 2})$ that has been produced by the first part of {\sf TreeBiSection}.   
  Let $n=|t|$ and $\sigma$ be number of different node labels that appear in $t$. 
  For the modified derivation tree $\mc D_t^*$ we have the following:
 \begin{itemize}
 \item $\mc D_t^*$ is a binary tree with $n$ leaves and hence has $2n-1$ nodes.
 \item There are $\sigma+3$ possible node labels, namely $1,2,3$ and those appearing in $t$.
 \item  $\mc D_t^*$ is $(1/7)$-balanced by \eqref{eq-balanced-splitting}. 
  If we have two successive nodes in $\mc D_t^*$, then we split at one of the two nodes according 
  to  \eqref{eq-balanced-splitting}. 
  Now, assume that we split at node $v$ according to
   \eqref{eq-balanced-splitting}. 
   Let $v_1$ and $v_2$ be the children of $v$,  let 
  $n_i$ be the leaf size of $v_i$, and let $n = n_1+n_2$ be the leaf size of $v$. 
  We get $\frac{1}{8} n  \leq n_1 \leq \frac{3}{4} n$ and $\frac{1}{4} n  \leq n_2 \leq \frac{7}{8} n$
  (or vice versa). Hence,  $n_1 \geq \frac{1}{8} n \geq \frac{1}{7} n_2$
  and $n_2 \geq \frac{1}{4} n \geq \frac{1}{3} n_1$. 
  \end{itemize}
  Recall that the nodes of the dag of of $\mc D_t^*$ are the nonterminals of the TSLP
  produced by {\sf TreeBiSection} and that this TSLP is in Chomsky normal form.
  Moreover, recall that the depth of $\mc D_t^*$ is in $O(\log n)$.
  Hence, with Lemma~\ref{lemma-space-time-treebisection} we get:
  
 \begin{corollary}\label{theorem:nlogN}
  For a  tree from $\mc T(\mc F_{\leq 2})$ of size $n$ with $\sigma$ different node labels, {\sf TreeBiSection} produces
  a TSLP in Chomsky normal form of size $O \big( \frac{n}{\log_\sigma n}\big)$ and depth $O(\log n)$. Every nonterminal
  of that TSLP has rank at most $3$, and the algorithm can be implemented in logspace and, alternatively, in time $O(n \cdot \log n)$.
 \end{corollary}
  In particular, for an unlabelled tree of size $n$ we obtain a TSLP of size $O \big( \frac{n}{\log n}\big)$.

\subsection{Extension to trees of larger degree} \label{sec-large-degree}

If the input tree $t$ has nodes with many children, then we cannot expect good compression by TSLPs.
The extreme case is $t_n = f_n(a,\ldots, a)$ where $f_n$ is a symbol of rank $n$. Hence, $|t_n| = n+1$
and every TSLP for $t_n$ has size $\Omega(n)$.
On the other hand, for trees, where the maximal rank is bounded by a constant $r \geq 1$, we can easily generalize
{\sf TreeBiSection}. Lemma~\ref{lemma:r-lemma} allows to find a splitting node $v$ satisfying 
\begin{equation} \label{eq-balanced-splitting-r}
\frac{1}{2(r+2)} \cdot |t|  \; \le \; |t[v]| \; \le \; \frac{r+1}{r+2} \cdot |t| .
\end{equation}
The maximal arity of of nodes also affects the arity of patterns: we allow patterns of rank up to $r$.
Now assume that $t(x_1,\ldots, x_{r+1})$ is a valid pattern of rank $r+1$, where 
$r$ is again the maximal number of children of a node. Then we find a subtree containing $k$
parameters, where $2 \leq k \leq r$: Take a smallest subtree that contains at least two parameters.
Since the root node of that subtree has at most $r$ children, and every proper subtree contains
at most one parameter (due to the minimality of the subtree), this subtree contains at most $r$
parameters. By taking the root of that subtree as the splitting node, we obtain two valid
patterns with at most $r$ parameters each. 
Hence, we have to change {\sf TreeBiSection} in the following way:
\begin{itemize}
\item As long as the number of parameters of the tree is at most $r$, we choose the splitting node according to Lemma~\ref{lemma:r-lemma}.
\item If the number of parameters is $r+1$ (note that in each splitting step, the number of parameters increases by at most $1$), then we choose
the splitting node such that the two resulting fragments have rank at most $r$. 
\end{itemize}
As before, this guarantees that in every second splitting step we split in a balanced way. But  the balance factor $\beta$ from Section~\ref{sec-dag-size}
now depends on $r$. More precisely, if in the modified derivation tree ${\mc D}^*_t$ we have a node $v$ with children $v_1$ and $v_2$ of leaf size $n_1$ and $n_2$,
respectively, and this node corresponds to a splitting satisfying \eqref{eq-balanced-splitting-r}, then 
we get
\begin{eqnarray*}
\frac{1}{2(r+2)} \cdot n   & \le  & n_1 \;  \; \le  \;\; \frac{r+1}{r+2} \cdot n , \\
\frac{1}{r+2} \cdot n \ = \ \left(1 - \frac{r+1}{r+2}\right) \cdot n & \leq & n_2  \; \; \leq\; \; \left(1- \frac{1}{2(r+2)}\right) \cdot n \ = \ \frac{2r+3}{2(r+2)} \cdot n 
\end{eqnarray*}
or vice versa.
This implies
\begin{eqnarray*}
   n_1 & \geq & \frac{1}{2(r+2)} \cdot n  \ \geq \ \frac{1}{2r+3} \cdot n_2 ,\\
   n_2 & \geq & \frac{1}{r+2} \cdot n  \ \geq \ \frac{1}{r+1} \cdot n_1 .
\end{eqnarray*}
Hence, the modified derivation tree becomes $\beta$-balanced for $\beta = 1/(2r+3)$.
Moreover, the label alphabet of the  modified derivation tree now has size $\sigma+r+1$ (since the TSLP produced in the first step
has nonterminals of rank at most $r+1$).
The proof of Theorem~\ref{thm-small-dag} yields the following bound on the dag of the modified derivation tree and hence the size of the final TSLP:
$$
O\left( \frac{n}{\log_{\sigma+r}(n)} \cdot \log_{1+\frac{1}{2r+3}} (2r+3) \right)  = O\left( \frac{n}{\log_{\sigma+r}(n)} \cdot \frac{\log (2r+3)}{\log(1+\frac{1}{2r+3})} \right)
$$
Note that $\log(1+x) \geq x$ for $0 \leq x \leq 1$. Hence, we can simplify the bound to
$$
O\left( \frac{n \cdot \log(\sigma+r) \cdot   r \cdot  \log r}{\log n} \right) .
$$
By Lemma~\ref{lemma:r-lemma},
the depth of the produced TSLP can be bounded by $2 \cdot d$, where $d$ is any number
that satisfies 
$$
n \cdot \bigg( \frac{r+1}{r+2} \bigg)^d \leq 1 .
$$
Hence, we can bound the depth by
$$
2 \cdot \bigg\lceil  \frac{\log n}{\log(1 + \frac{1}{r+1})}   \bigg\rceil  \leq  
2 \cdot \lceil (r+1) \cdot  \log n \rceil  \in O(r \cdot \log n).
$$


\begin{theorem} \label{thm-treebisection-non-constant-rank}
For a  tree of size $n$ with $\sigma$ different node labels, each of rank at most $r$,
{\sf TreeBiSection} produces a TSLP in Chomsky normal form of size $O\big( \frac{n \cdot \log(\sigma+r) \cdot   r \cdot  \log r}{\log n} \big)$ and depth $O(r \cdot \log n)$.
Every nonterminal  of that TSLP has rank at most $r+1$.
\end{theorem}
For the running time we obtain the following bound:

\begin{theorem} \label{thm-treebisection-non-constant-rank-time}
{\sf TreeBiSection} can be implemented such that it works in time $O(r \cdot n \cdot \log n)$ 
for a  tree of size $n$, where each symbol has rank at most $r$.
\end{theorem}

\begin{proof}
{\sf TreeBiSection} makes $O(r \cdot \log n)$ iterations of the while loop (this is the same bound as for the depth of the TSLP)
and each iteration takes time $O(n)$. 
To see the latter, our internal representation of trees with parameters from Section~\ref{sec-trees} is important.
Using this representation, we can still compute the split node in a right-hand side $s$ from $P_{\mathrm{temp}}$ in
time $O(|s|)$: We fist compute for every non-parameter node $v$ of $s$ (i) the size of the subtree rooted 
at $v$ (as usual, excluding parameters) and (ii) the number of parameters below $v$. This is possible in time 
$O(|s|)$ using a straightforward bottom-up computation. Using these size informations, we can compute the split
node in $s$ in time $O(|s|)$ for both cases (number of parameters in $s$ is $r+1$ or smaller than $r+1$) by searching from
the root downwards.
\end{proof}
In particular, if $r$ is bounded by a constant, {\sf TreeBiSection} computes a 
TSLP of size $O\big(\frac{n}{\log_\sigma n}\big)$ and depth $O(\log n)$ in time $O(n \cdot \log n)$.
Moreover, our logspace implementation of {\sf TreeBiSection} 
(see Lemma~\ref{lemma-space-time-treebisection}) directly
generalizes to the case of a constant rank. 


On the other hand, for unranked trees where the number of children of a node
is arbitrary and not determined by the node label (which is the standard tree model in XML) all this fails:
{\sf TreeBiSection} only yields TSLPs of size $\Theta(n)$ and this is unavoidable as shown by the example from the beginning of this section.
Moreover, the logspace implementation from Section~\ref{sec-time-space} no longer works since nonterminals
have rank at most $r+1$ and we cannot store anymore the pattern derived from a nonterminal in space $O(\log n)$
(we have to store $r+1$ many nodes in the tree). 

Fortunately there is a simple workaround for all these problems: An unranked
tree can be transformed into a binary tree of the same size using the well known first-child next-sibling
encoding \cite{MLMN13,Knuth68}. Then, one can simply apply {\sf TreeBiSection}  to this encoding
to get in logspace and time $O(n \cdot \log n)$ a TSLP of size  $O\big( \frac{n}{\log_\sigma n}\big)$. 

For the problem of traversing a compressed unranked tree $t$ (which is addressed in \cite{BilleGLW13} for top 
dags) another (equally well known) encoding is more favorable. Let $c(t)$ be a compressed representation
(e.g., a TSLP or a top dag) of $t$. The goal is to represent $t$ in space $O(|c(t)|)$ such that one can efficiently navigate from
a node to (i) its parent node, (ii) its first child, (iii) its next sibling, and (iv) its previous sibling (if they exist).
For top dags \cite{BilleGLW13}, it was shown that a single navigation step can be done in time $O(\log |t|)$.
Using the right binary encoding, we can prove the same result for TSLPs:
Let  $r$ be the maximal rank of a node of the unranked tree $t$.
We define the binary encoding $\mathrm{bin}(t)$ by adding for every node $v$ of rank $s \leq r$
a binary tree of depth $\lceil \log s \rceil$ with $s$ many leaves,
whose root is $v$ and whose leaves are the 
children of $v$.  This introduces at most $2s$ many new binary nodes, which are labelled by a new symbol.  
We get $|\mathrm{bin}(t)| \leq 3 |t|$. In particular, we obtain a TSLP of size $O\big( \frac{n}{\log_\sigma n}\big)$
for $\mathrm{bin}(t)$, where $n = |t|$ and $\sigma$ is the number of different node labels.
Note that a traversal step in the initial tree $t$ (going to the parent node,
first child, next sibling, or previous sibling)
can be simulated by $O(\log r)$ many traversal
steps in $\mathrm{bin}(t)$ (going to the parent node, left child, or right child).
But for a binary tree $s$, it was recently shown that a TSLP $\mc G$ 
for $s$ can be represented in space $O(|\mc G|)$ such that a single traversal
step takes time $O(1)$ \cite{Loh14}.\footnote{This generalizes a corresponding result for strings~\cite{GasieniecKPS05}.}
Hence, we can navigate in $t$ in time $O(\log  r) \leq O(\log |t|)$.

\subsection{Linear Time TSLP Construction}\label{sec:lin_time}

  Recall that {\sf TreeBiSection} works in time $O(n \cdot \log n)$ for a tree of size $n$.
  In this section we present a linear time algorithm {\sf BU-Shrink} (for bottom-up shrink) that also constructs
  a TSLP of size $O\big(\frac{n}{\log_\sigma n}\big)$ and
  depth $O(\log n)$ for a given tree of size $n$ with $\sigma$ many node labels
  of constant rank.
  The basic idea of {\sf BU-Shrink} is to merge in a bottom-up way nodes
  of the tree to \emph{patterns} of size roughly $k$, 
  where $k$ is defined later.
  This is a bottom-up computation in the sense that we begin with individual nodes and gradually merge them into larger fragments
  (the term ``bottom-up'' should not be understood in the sense that the computation is done from the leaves of the tree towards the root).
  The dag of the small trees represented by the patterns then yields the compression.

  For a valid pattern $p$ of rank $r$ we define
  the weight of $p$ as $|p|+r$. This is the total number of nodes in $p$ including those
  nodes that are labelled with a parameter (which are not counted in the size $|p|$ of $p$).
  A \emph{pattern tree} is a tree in which the labels of the tree are valid patterns. 
  If a node $v$ is labelled with the valid pattern $p$ and $\rank(p) = d$, then we 
  require that $v$ has $d$ children in the pattern tree. 
  For convenience, {\sf BU-Shrink} also stores in every node the weight of the corresponding 
  pattern. For a node $v$, we denote by $p_v$ its pattern and by $w(v)$ the weight of $p_v$.
  
  Let us fix a number $k \geq 1$ that will be specified latter.
  Given a tree $t = (V, \lambda)$ of size $n$ such that all node labels in $t$ are of rank at most $r$, {\sf BU-Shrink} first creates a pattern tree
  (which will be also denoted with $t$)
  by replacing every label $f \in \mc F_d$ by the valid pattern $f(x_1,\dots, x_d)$ of weight $d+1$.
  Note that the the parameters in these patterns correspond to the edges of the tree $t$.
  We will keep this invariant during the algorithm (which will shrink the tree $t$). 
  Hence, the total number of all parameter occurrences in the patterns that appear as labels in the current pattern tree $t$
  will be always the number of nodes of the current tree $t$ minus $1$. This allows us to ignore the cost of 
  handling parameters for the running time of the algorithm.
  
  {\sf BU-Shrink} then creates a queue $Q$ that contains references to all nodes of $t$ having at most one child (excluding the root node) 
  in an arbitrary ordering. During the run of the algorithm, the queue $Q$ will only contain references to non-root nodes of the current tree $t$ that
  have at most one child (but $Q$ may not contain references to all such nodes).
  For each node $v$ of the queue we proceed as follows.
  Let $v$ be the $i^\mathrm{th}$ child of its parent node $u$.
  If $w(v) > k$ or $w(u) > k$, we simply remove $v$ from $Q$ and proceed.
  Otherwise we merge the node $v$ into the node $u$.
  More precisely, we delete the node $v$, and 
  set the $i^\mathrm{th}$ child of $u$ to the unique 
  child of $v$ if it exists (otherwise, $u$ loses its $i^\mathrm{th}$ child). 
  The pattern $p_{u}$ is modified by replacing the parameter at the position of the $i^\mathrm{th}$
  child by the pattern $p_v$ and re-enumerating all parameters 
  to get a valid pattern.
  We also set the weight $w(u)$ to $w(v)+w(u)-1$ (which is the weight of the new pattern $p_{u}$).
  Note that in this way both the number of edges of $t$ and the total number of parameter occurrences in all patterns
  decreases by $1$ and so these two sizes stay equal.
  For example, let $u$ be a node with  $p_{u} = f(x_1,x_2)$ and let $v$ 
  be its second child with $p_v = g(x_1)$.
  Then the merged pattern becomes
  $f(x_1,g(x_2))$, and its weight is $4$.
  If the node $u$ has at most one child after the merging
  and its weight is at most $k$, then we add $u$ to $Q$ (if it is not already in the queue).
  We do this until the queue is empty. Note that every pattern appearing in the final pattern tree has rank at most $r$
  (the maximal rank in the initial tree).
  
  Now consider the forest consisting of all patterns appearing in the resulting final pattern tree.  
  We construct the dag of this forest, which yields grammar rules for each pattern with shared nonterminals.
  The dag of a forest (i.e., a disjoint union of trees) is constructed in the same way as for a single tree.
  This dag has for every subtree appearing in the forest exactly one node. The parameters $x_1, x_2, \ldots, x_r$
  that appear in the patterns are treated as ordinary constants when constructing the dag. 
  As usual, the dag can be viewed as a TSLP, where the nodes of the dag are the  nonterminals.
  In this way we obtain a TSLP in which each pattern is derived by a nonterminal of the same rank. 
  Finally, we add to the TSLP the start rule $S \to s$,
  where $s$ is obtained from the pattern tree by labelling each node $v$ with the unique nonterminal
  $A$ such that $A$ derives $p(v)$.
  Algorithm~\ref{algorithm:new} shows a pseudocode for {\sf BU-Shrink}.

\begin{algorithm}[t]
\SetKwComment{Comment}{}{}
\SetKwInOut{Input}{input}
\SetKwInOut{Sub}{methods}
\Input{binary tree $t=(V,\lambda)$, number $k \leq |t|$}
$Q := \emptyset$ \\
\ForEach{$v \in V$}
  {     let $f = \lambda(v) \in \mc F_d$ be the label of node $v$ \\
	$w(v) := 1+d$  \hfill  \Comment{(the weight of  node $v$)}  
	$p_v := f(x_1, \dots, x_d)$   \hfill \Comment{(the pattern stored in node $v$)} 
	\If{$d \le 1$ and $v$ is not the root}
	  {
		$Q :=Q \cup \{ v \}$
	  }
  }
\While{$Q \neq\emptyset$}
  {
	choose arbitrary node $v \in Q$ and set $Q:= Q \setminus \{v\}$ \\
	 let $u$ be the parent node of $v$ \\
	\If{ $w(v) \leq k$ {\bf and} $w(u) \leq k$}
	  {     $d := \rank(p_v)$; $e := \rank(p_u)$  \\
	        let $v$ be the $i^{\text{th}}$ child of $u$ \\
		$w(u) := w(u) + w(v)-1$\\
		$p_u := p_u(x_1, \ldots, x_{i-1}, p_v(x_i, \ldots, x_{i+d-1}), x_{i+d}, \ldots, x_{d+e-1})$\\
		\If{ $v$ has a (necessarily unique) child $v'$}{set $v'$ to the $i^{\text{th}}$ child of $u$}
		delete node $v$ \\
		\If{ $d+e-1 \le 1$ and $w(u) \leq k$}
		  {
			$Q :=Q \cup \{ u \}$
		  }
	  }
  }
  $P := \emptyset$\\
  compute the minimal dag for the forest consisting of all patterns $p_v$ for $v$ a node of $t$ \\
  \ForEach{node $v$ of $t$}{
     create a fresh nonterminal $A_v$ of rank $d := \rank(p_v)$ \\
     $P := P \cup \{ A_v(x_1, \ldots, x_d) \to p_v(x_1, \ldots, x_d) \}$ \\
     $\lambda(v) := A_v$ \hfill \Comment{(the new label of node $v$)}
     }
  \Return TSLP $(S,P \cup \{ S \to t \})$
\caption{{\sf BU-Shrink}$(t,k)$  \label{algorithm:new}}
\end{algorithm}

  
 \begin{example}
   Consider the pattern tree depicted in Figure~\ref{fig:exLinTime}.
   Assuming no further mergings are done, the final TSLP is
   $S \to A(B(C),B(B(C))),\,A(x,y) \to f(g(x),y),\, B(y) \to f(C,y),\, C \to g(a)$.
 \end{example}

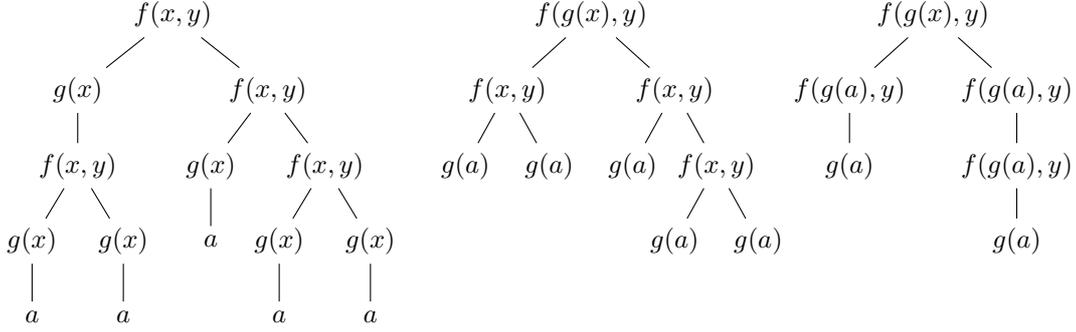
\begin{figure}[t]

\centering

\begin{tikzpicture}[scale=1,auto,swap,level distance=10mm]

\tikzset{level 1/.style={sibling distance=25mm}}
\tikzset{level 2/.style={sibling distance=15mm}}
\tikzset{level 3/.style={sibling distance=12mm}}  
\tikzset{level 4/.style={sibling distance=4mm}}
\node (eps)  at (2,4) {$f(x,y)$} {
    child {node (0) {$g(x)$}
      child {node (00) {$f(x,y)$}
        child {node (000) {$g(x)$}
          child {node (0000) {$a$}}
        }
        child {node (001) {$g(x)$}
          child {node (0010) {$a$}}
        }
      }
    }
    child {node (1) {$f(x,y)$}
      child {node (10) {$g(x)$}
        child {node (100) {$a$}}
      }
      child {node (11) {$f(x,y)$}
        child {node (000) {$g(x)$}
          child {node (0000) {$a$}}
        }
        child {node (000) {$g(x)$}
          child {node (0000) {$a$}}
        }
      }
    }
 };
\tikzset{level 1/.style={sibling distance=22mm}}
\tikzset{level 2/.style={sibling distance=11mm}}
\tikzset{level 3/.style={sibling distance=11mm}}
\node (b-eps)  at (7.5,4) {$f(g(x),y)$} {
	child {node (b-0) {$f(x,y)$}
	  child {node (b-00) {$g(a)$}}
	  child {node (b-01) {$g(a)$}}
	}
    child {node (b-1) {$f(x,y)$}
      child {node (b-10) {$g(a)$}}
      child {node (b-11) {$f(x,y)$}
        child {node (b-110) {$g(a)$}}
        child {node (b-111) {$g(a)$}}
      }
    }
 };

\node (c-eps)  at (12,4) {$f(g(x),y)$} {
	child {node (c-0) {$f(g(a),y)$}
	  child {node (c-00) {$g(a)$}}
	}
    child {node (c-1) {$f(g(a),y)$}
      child {node (c-10) {$f(g(a),y)$}
        child {node (c-100) {$g(a)$}}
      }
    }
 }; 

 \end{tikzpicture}
  \caption{When compressing, we start with the pattern tree left 
  (the weights are omitted to improve readability). 
  The second tree and third tree depict possible intermediate steps.}
 \label{fig:exLinTime}
\end{figure} 
  
It is easy to see that {\sf BU-Shrink} runs in time $O(n)$ for a tree of size $n$.
First of all, the number of mergings is bounded by $n$, since each merging reduces
the number of nodes of the pattern tree by one. Moreover, if a node is removed from $Q$
(because its weight or the weight of its parent node is larger than $k$)
then it will never be added to $Q$ again (since weights are never reduced). 
A single merging step needs only
a constant number of pointer operations and a single addition (for the weights).
For this, it is important that we do not copy patterns, when the new pattern (for the node
$u$ in the above description) is constructed. This size of the forest, for which we construct
the dag has size $O(n)$: The number of non-parameter nodes is exactly $n$,
and the number of parameters is at most $n-1$: initially the forest
has $n-1$ parameters (as there is a parameter for each node except the root)
and during {\sf BU-Shrink} we can only decrease the total amount of parameters.

 Let us now analyze the size of the constructed TSLP.
 In the following, let $t$ be the input tree of size $n$ and let $r$ be the maximal rank of a label
 in $t$. Let $\Sigma = \labels(t)$ and $\sigma = |\Sigma|$.

 \begin{lemma}\label{lemma:size_pattern_tree}
 Let $t_p$ be the pattern tree resulting from {\sf BU-Shrink}.
 Then $|t_p| \le \frac{4 \cdot r \cdot n}{k}+2$.
 \end{lemma}
 
 \begin{proof} Let us first assume that $r \geq 2$.
  The number of non-root nodes in $t_p$ of arity at most one is at least $|t_p|/2-1$.
  For each of those nodes, either the node itself or the parent node has weight
  at least $k$. We now map each non-root node of arity at most one to a node of weight at least $k$:
  Let $u$ (which is not the root) have arity at most one.
  If the weight of $u$ is at least $k$, we map $u$ to itself, otherwise we map $u$ to its parent node,
  which then must be of weight at least $k$.
  Note that at most $r$ nodes are mapped to a fixed node $v$: If $v$ has arity at most one, then $v$ and its child (if it exists) can be mapped to $v$;
  if $v$ has arity greater than one then only its children can be mapped to $v$, and $v$ has at most $r$ children.
  Therefore there must exist at least $\frac{|t_p|-2}{2r}$ many nodes of weight at least $k$.
  Because the sum of all weights in $t_p$ is at most $2n$, this yields
  \[
   \frac{|t_p|-2}{2r} \cdot k \le 2n,
  \]
  which proves the lemma for the case $r \geq 2$. The case $r=1$ can be proved in the same way:
  Clearly, the number of non-root nodes of arity at most one is $|t_p|-1$ and at most $2$ nodes are mapped to 
  a fixed node of weight at least $k$ (by the above argument). Hence,  
  there exist at least $\frac{|t_p|-1}{2}  = \frac{|t_p|-1}{2r}$ many nodes of weight at least $k$.
 \end{proof}
 Note that each node in the final pattern tree has weight at most  $2 k$ since {\sf BU-Shrink} only merges nodes of weight at most $k$.
 By Lemma~\ref{lemma-counting}  the number of different patterns in $\mc T(\Sigma \cup \{x_1, \ldots, x_r\})$ 
 of weight at most $2k$ is bounded by $\frac{4}{3} (4 (\sigma+r))^{2k} \leq d^{k}$ for $d = (6(\sigma+r))^2$. 
 Hence, the size of the dag constructed from the patterns is bounded by $d^k$.
Adding the size of the start rule, i.e. the size of the resulting pattern tree (Lemma~\ref{lemma:size_pattern_tree}) we get the following bound for the 
constructed TSLP:
$$
d^k + \frac{4 \cdot r \cdot n}{k} + 2.
$$
Let us now set $k = \frac{1}{2} \cdot \log_d n$.
We get the following bound for the constructed TSLP:
 $$
 d^{\frac{1}{2} \cdot \log_d n}  + \frac{8 \cdot n \cdot r}{\log_d n} + 2 =  \sqrt{ n} + O\left( \frac{n \cdot r}{\log_d n}\right) = O\left( \frac{n \cdot r}{\log_{\sigma+r} n}\right)
 = O\left(\frac{n \cdot \log(\sigma+r) \cdot r}{\log n}\right).
 $$
 \begin{theorem}
{\sf BU-Shrink} computes for a given tree of size $n$ with $\sigma$ many node labels of rank at most $r$ in time $O(n)$
a TSLP of size $O\big( \frac{n \cdot \log(\sigma+r) \cdot r}{\log n}\big)$. Every nonterminal of that TSLP has rank at most $r$.
\end{theorem}
Clearly, if $r$ is bounded by a constant, we obtain the bound
$O\big(\frac{n}{\log_\sigma n}\big)$. On the other hand, already for $r \in \Omega(\log n)$ the bound $O\big( \frac{n \cdot r}{\log_{r+\sigma} n}\big)$
is in $\omega(n)$. But note that the size of the TSLP produced by {\sf BU-Shrink} can never exceed  $n$.

 \paragraph{Combining TreeBiSection and BU-Shrink to achieve logarithmic grammar depth.}

Recall that {\sf TreeBiSection} produces TSLPs in Chomsky normal form of logarithmic depth, which will be important
in the next section. Clearly, the TSLPs produced by {\sf BU-Shrink} are not in Chomsky normal form. To get a TSLP
in Chomsky normal form we have to further partition the right-hand sides of the TSLP. Let us assume in this 
section that the maximal rank of symbols appearing in the input tree is bounded by a constant. Hence, 
{\sf BU-Shrink}  produces for an input tree $t$  of size $n$ with $\sigma$ many node labels
a TSLP of size $O\big(\frac{n}{\log_\sigma n}\big)$.  The weight and hence also the depth 
of the patterns that appear in $t_p$ is $O(\log_d n) = O(\log_\sigma n) \leq O(\log n)$.
The productions that arise from the dag of the forest of all patterns have the form $A \to f(A_1, \ldots, A_r)$ 
and $A \to x$ where $f$ is a node label of the input tree, $r$ is bounded by a constant, and $x$ is one of the parameters (recall that
in the dag construction, we consider  the parameters appearing in the patterns as ordinary constants).
Productions $A \to x$ are eliminated by replacing all occurrences of $A$ in a right-hand side by the parameter $x$.
The resulting productions can be split into Chomsky normal form productions in an arbitrary way (see also \cite{LoMaSS12}).
This increases the size and depth only by the maximal rank of node labels, which is a constant in our consideration.

Recall that for the start rule  $S \to s$, where $s$ is the tree returned by {\sf BU-Shrink},
the tree $s$ has size $O\big(\frac{n}{\log_\sigma n}\big) = O\big(\frac{n \cdot \log \sigma}{\log n}\big)$.
  We want to apply {\sf TreeBiSection} to balance the tree $s$.
  But we cannot use it directly because the resulting running time would not be linear if $\sigma$ is not a constant:
  Since {\sf TreeBiSection}  needs time $O(|s| \log |s|)$ on trees of constant rank (see the paragraph after the proof of Theorem~\ref{thm-treebisection-non-constant-rank-time}), 
  this yields the time bound
  \[
	  O\left(\frac{n \cdot \log \sigma}{\log n}  \cdot \log \left(\frac{n \cdot \log \sigma}{\log n}\right)\right) = O\left(\frac{n \cdot \log \sigma}{\log n} \cdot \left( \log n + \log \log \sigma - \log \log n \right) \right) = O(n \log \sigma).
  \]
  To eliminate the factor $\log \sigma$, we apply {\sf BU-Shrink}  again to $s$ with $k = \log \sigma \leq \log n$. Note that the maximal rank
  in $s$ is still bounded by a constant (the same constant as for the input tree).
  By Lemma \ref{lemma:size_pattern_tree} this yields in time $O(|s|) \leq O(n)$
  a tree $s'$ of size $O\big(\frac{n}{\log n}\big)$  on which we may now use {\sf TreeBiSection} to get 
  a TSLP for $s'$ in Chomsky normal form of size $O(|s'|) = O\big(\frac{n}{\log n}\big)$ 
  (note that every node of $s'$ may be labelled with a different symbol, in which case
  {\sf TreeBiSection} cannot achieve any compression for $s'$, when we count the size in bits) and depth $O(\log |s'|) = O(\log n)$. 
  Moreover, the running time of  {\sf TreeBiSection} on $s'$ is 
   $$
   O(|s'| \cdot \log |s'|) = O\left(\frac{n}{\log n} \log \left(\frac{n}{\log n}\right)\right) = O\left(\frac{n}{\log n} \cdot \left( \log n - \log \log n \right) \right) = O(n).
   $$
   Let us recall this combined algorithm {\sf BU-Shrink+TreeBiSection}.
   
   \begin{theorem} \label{thm-linear-time-log-depth}
  {\sf BU-Shrink+TreeBiSection} computes for a given tree $t$ of size $n$ with $\sigma$ many node labels of constant rank each in time $O(n)$
 a TSLP in Chomsky normal form of size $O\big( \frac{n}{\log_\sigma n}\big)$ and depth $O(\log n)$. The rank of every nonterminal of that TSLP is bounded by the maximal
 rank of a node label in $t$ (a constant).
  \end{theorem}

 \section{Arithmetical Circuits} \label{sec-circuits}
  
In this section, we present our main application of Corollary~\ref{theorem:nlogN} and Theorem~\ref{thm-linear-time-log-depth}.
Let $\mc S= (S,+, \cdot)$ be a (not necessarily commutative) semiring. 
Thus, $(S,+)$ is a commutative monoid with identity element $0$, $(S,\cdot)$ is a monoid
with identity element $1$, and $\cdot$ left and right distributes over $+$.

We use the standard notation of arithmetical formulas and
circuits over $\mc S$: 
An {\em arithmetical formula} is just a labelled binary tree where internal nodes are labelled with the semiring operations $+$ and $\cdot$,
and leaf nodes are labelled with variables $y_1, y_2, \ldots$ or the constants $0$ and $1$.
An {\em arithmetical circuit} is a (not necessarily minimal) dag whose internal nodes are labelled with $+$ and $\cdot$ and 
whose leaf nodes are labelled with variables or the constants $0$ and $1$.  
The \emph{depth} of a circuit is the length of a longest path from the root node to a leaf.
An arithmetical circuit evaluates to a multivariate noncommutative polynomial $p(y_1, \ldots, y_n)$ over $\mc S$, where
$y_1, \ldots, y_n$ are the variables occurring at the leaf nodes.
Two arithmetical circuits are equivalent if they evaluate to the same polynomial.

Brent \cite{Brent74} has shown that every arithmetical formula of size $n$ over a commutative ring can be transformed
into an equivalent circuit of depth $O(\log n)$ and size $O(n)$ (the proof easily generalizes to semirings). 
By first constructing a TSLP of size $O\big(\frac{n \cdot \log m}{\log n}\big)$, where $m$ is the number of different variables in the formula,
and then transforming this TSLP into a circuit, we will
refine the size bound to $O\big(\frac{n \cdot \log m}{\log n}\big)$.
Moreover, by Corollary~\ref{theorem:nlogN} (Theorem~\ref{thm-linear-time-log-depth}, respectively) this 
conversion can be done in logspace (linear time, respectively).

In the following, we consider TSLPs, whose  terminal alphabet consists of the binary symbols $+$ and $\cdot$ 
and the constant symbols $0,1, y_1, \ldots, y_m$ for some $m$. Let us denote this terminal alphabet with $\Sigma_m$.
For our formula-to-circuit conversion, it will be important to work with monadic TSLPs, i.e., TSLPs where every nonterminal 
has rank at most one.

\begin{lemma} \label{lemma-monadic-logspace}
From a given tree $t \in \mc T(\Sigma_m)$ of size $n$
one can construct in logspace (linear time, respectively) a monadic TSLP $\mc H$ of size $O\big(\frac{n \cdot \log m}{\log n}\big)$ and depth $O(\log n)$
with $\val(\mc H) = t$ and such that all productions are
of the following forms:
\begin{itemize}
\item $A \to B(C)$ for $A, C \in \mc N_0$, $B \in \mc N_1$, 
\item $A(x) \to B(C(x))$ for $A,B,C \in \mc N_1$,
\item $A \to f(B,C)$ for $f \in \{+, \cdot \}$, $A,B,C \in \mc N_0$,
\item $A(x) \to f(x,B)$, $A(x) \to f(B,x)$  for $f \in \{+, \cdot \}$, $A \in \mc N_1$, $B \in \mc N_0$, 
\item $A \to a$ for $a \in \{0,1, y_1, \ldots, y_m\}$, $A  \in \mc N_0$,
\item $A(x) \to B(x)$, $A(x) \to x$ for $A,B \in \mc N_1$.
\end{itemize}

\end{lemma}

\begin{proof}
The linear time version is an immediate consequence of Theorem~\ref{theo-monadic} and Theorem~\ref{thm-linear-time-log-depth}.\footnote{Note that productions of the
form $A(x) \to B(x)$ and $A(x) \to x$ do not appear in Theorem~\ref{theo-monadic}. We allow them in the lemma, since they make the logspace part of the lemma easier to show
and do not pose a problem in the remaining part of this section.}
It remains show the logspace version. 
We first apply {\sf TreeBiSection} (Corollary~\ref{theorem:nlogN}) to get in logspace a TSLP $\mc G$ in Chomsky normal form of size $O\big(\frac{n \cdot \log m}{\log n}\big)$ and depth $O(\log n)$
with $\val(\mc G) = t$. Note that every nonterminal of $\mc G$ has rank $3$. Moreover, for every nonterminal $A$ of rank $k \leq 3$, {\sf TreeBiSection} 
computes $k+1$ nodes $v_0, v_1, \ldots, v_k$ of $t$ that represent the pattern $\val_{\mc G}(A)$: $v_0$ is the root node
of an occurrence of $\val_{\mc G}(A)$ in $t$ and $v_i$ ($1 \leq i \leq k$) is the node of the occurrence to which the parameter $x_i$ is mapped, see also the proof
of Lemma~\ref{lemma-space-time-treebisection}.
We can assume that for every nonterminal $A$ of rank $k$ this tuple $s_A$ has been computed.

We basically show that the construction from  \cite{LoMaSS12}, which makes a TSLP monadic, works in logspace
if all nonterminals and terminals of the input TSLP have constant rank.\footnote{We only 
consider the case that nonterminals have rank at most three and terminals have rank zero or two, which is the case
we need, but the general case, where all nonterminals and all terminals of the input TSLP have constant rank could be 
handled in a similar way in logspace.}
For a nonterminal $A$ of rank $3$ with $s_A = (v_0, v_1,v_2,v_3)$, the pattern $\val_{\mc G}(A)$ has two possible branching structures,
see Figure~\ref{fig:threevar}. By computing the paths from the three nodes $v_1, v_2, v_3$ up to $v_0$,
we can compute in logspace, which of the two branching structures $\val_{\mc G}(A)$ has. Moreover, we can compute
the two binary symbols $f_1, f_2 \in \{ +, \cdot\}$ at which the three paths that go from $v_1$, $v_2$, and $v_3$, respectively,
up to $v_0$ meet. We finally, associate with each of the five dashed edges in Figure~\ref{fig:threevar} a fresh unary nonterminal $A_i$
($0 \leq i \leq 4$) of the TSLP $\mc H$. In this way we can built up in logspace what is called the skeleton tree for $A$. It is one of the following two trees,
depending on the branching structure of $\val_{\mc G}(A)$, see also Figure~\ref{fig-branching-structures}:
$$
A_0(f_1(A_1(f_2(A_2(x_1), A_3(x_1))), A_4(x_3)))  \qquad A_0(f_1(A_1(x_1),A_2(f_2(A_3(x_2), A_4(x_3)))))
$$
\begin{figure}[t]
\tikzset{level 1/.style={sibling distance=32mm}}
\tikzset{level 2/.style={sibling distance=16mm}}
\tikzset{level 7/.style={sibling distance=8mm}}  
\tikzset{level 8/.style={sibling distance=4mm}}
\hspace*{\fill}
\begin{tikzpicture}[scale=1,auto,swap,level distance=8mm]
\node (eps) {$A_0$} 
  child {node {$f_1$}
    child {node {$A_1$}
      child {node {$f_2$}
        child {node {$A_2$}
          child {node {$x_1$}}
        }
        child {node {$A_3$}
          child {node {$x_2$}}
        }
      }
    }
    child {node {$A_4$}
      child {node {$x_3$}}  
    }
  }
;
{label fig one};
\end{tikzpicture}
\hspace*{\fill}
\begin{tikzpicture}[scale=1,auto,swap,level distance=8mm]
\node (eps) {$A_0$} 
  child {node {$f_1$}
    child {node {$A_1$}
      child {node {$x_1$}}  
    }
    child {node {$A_2$}
      child {node {$f_2$}
        child {node {$A_3$}
          child {node {$x_2$}}
        }
        child {node {$A_4$}
          child {node {$x_3$}}
        }
      }
    }
  }
;
{label fig one};
\end{tikzpicture}
\hspace*{\fill}
\caption{The two possible skeleton trees for a nonterminal $A$ of rank three}
 \label{fig-branching-structures}
\end{figure}
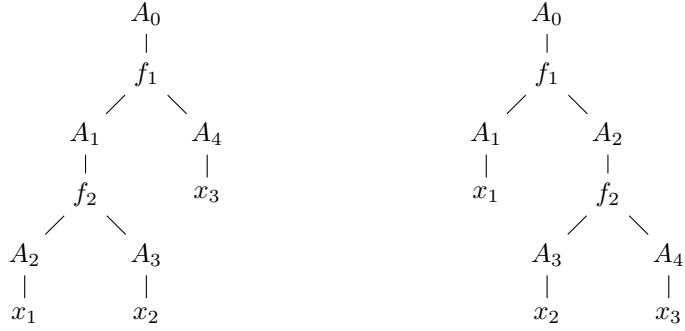 
For a nonterminal $A$ of rank two there is only a single branching structure and hence a single skeleton tree 
$A_0(f(A_1(x_1), A_2(x_2)))$ for $f \in \{+,\cdot\}$. Finally, for a nonterminal $A$ of rank at most one, the skeleton tree is $A$ itself
(this is in particular the case for the start nonterminal $S$, which will be also the start nonterminal of $\mc H$), 
and this nonterminal then belongs to $\mc H$ (nonterminals of $\mc G$ that
have rank larger than one do not belong to $\mc H$).
What remains is to construct in logspace productions for the nonterminals of $\mc H$ that allow to rewrite the skeleton tree of $A$ to $\val_{\mc G}(A)$.
For this, let us consider the productions of $\mc G$, whose right-hand sides have the form \eqref{nonterm-rules}  and \eqref{term-rules}. 
A production $A(x_1, \ldots, x_k) \to f(x_1, \ldots, x_k)$ with $k \leq 1$ 
is copied to $\mc H$. On the other hand, if $k = 2$, then $A$ does not belong to $\mc H$ and hence, we do not copy
the production to $\mc H$. Instead, we introduce the productions $A_i(x_1) \to x_1$ ($0 \leq i \leq 2$) for the three nonterminals
$A_0, A_1, A_2$ that appear in the skeleton tree of $A$.
 Now consider a production 
$$
A(x_1, \ldots, x_k) \to B(x_1, \ldots, x_{i-1}, C(x_i, \ldots, x_{i+l-1}), x_{i+l}, \ldots, x_k),
$$
where $k,l, k-l+1 \leq 3$.
We have constructed the skeleton trees $t_A,t_B, t_C$ for $A,B$, and $C$, respectively.
Consider the tree $t_B(x_1, \ldots, x_{i-1}, t_C(x_i, \ldots, x_{i+l-1}), x_{i+l}, \ldots, x_k)$.
We now introduce the productions for the nonterminals that appear in $t_A$ in such a way that $t_A(x_1, \ldots, x_k)$
can be rewritten to $t_B(x_1, \ldots, x_{i-1}, t_C(x_i, \ldots, x_{i+l-1}), x_{i+l}, \ldots, x_k)$. There are several cases depending
on $k,l$, and $i$. 
Let us  only consider two typical cases (all other cases can be dealt in a similar way):
If the production is $A(x_1,x_2,x_3) \to B(x_1, C(x_2,x_3))$ then the trees $t_A(x_1,x_2,x_3)$ and 
$t_B(x_1, t_C(x_2,x_3))$ are shown in Figure~\ref{fig-rewrite-skeleton-1}. Note that the skeleton tree 
$t_A(x_1,x_2,x_3)$ must be the right tree from Figure~\ref{fig-branching-structures}.
We add the following productions to $\mc H$:
$$
A_0(x) \to B_0(x),  \  \ A_1(x) \to B_1(x),  \  \ A_2(x) \to B_2(C_0(x)), \ \ A_3(x) \to C_1(x), \  \ A_4(x) \to C_2(x) .
$$
Let us also consider the case $A(x_1,x_2) \to B(x_1, x_2, C)$. The trees $t_A(x_1,x_2)$ and 
$t_B(x_1, x_2, t_C) = t_B(x_1, x_2, C)$ are shown in Figure~\ref{fig-rewrite-skeleton-2} (we assume 
that  the skeleton tree for $B$ is the left one from Figure~\ref{fig-branching-structures}).
We add the following productions to $\mc H$:
\begin{equation} \label{eq-skeleton}
A_0(x) \to B_0(f_1(B_1(x),B_4(C))),  \  \  A_1(x) \to B_2(x), \ \  A_2(x) \to B_3(x) . 
\end{equation}
Other cases can be dealt with similarly. In each case we write out a constant number of productions
that clearly can be produced by a logspace machine using the shape of the skeleton trees.
Correctness of the construction (i.e., $\val(\mc G) = \val(\mc H)$) follows from $\val_{\mc H}(t_A) = \val_{\mc G}(A)$,
which can be shown by a straightforward induction, see  \cite{LoMaSS12}.
Clearly, the size and depth of $\mc H$ is linearly related to the size and depth, respectively, of $\mc G$.
Finally, productions of the form $A(x) \to B(f(C(x),D(E)))$ (or similar
forms) as in \eqref{eq-skeleton} can be easily split in logspace into productions of the forms shown in the lemma. For instance,
$A(x) \to B(f(C(x),D(E)))$ is split into $A(x) \to B(F(x))$, $F(x) \to G(C(x))$, $G(x) \to f(x,H)$, $H \to D(E)$. Again, the size and 
depth of the the TSLP increases only by a linear factor.
\end{proof}


\begin{figure}[t]
\tikzset{level 1/.style={sibling distance=32mm}}
\tikzset{level 2/.style={sibling distance=16mm}}
\tikzset{level 7/.style={sibling distance=8mm}}  
\tikzset{level 8/.style={sibling distance=4mm}}
\hspace*{\fill}
\begin{tikzpicture}[scale=1,auto,swap, level 1/.style={level distance=8mm}, level 3/.style={level distance=16mm}, level 4/.style={level distance=8mm}]
\node (eps) {$A_0$} 
  child {node {$f_1$}
    child {node {$A_1$}
      child {node {$x_1$}}  
    }
    child {node {$A_2$}
      child {node (f2) {$f_2$}
        child {node {$A_3$}
          child {node {$x_2$}}
        }
        child {node {$A_4$}
          child {node {$x_3$}}
        }
      }
    }
  }
;
\end{tikzpicture}
\hspace*{\fill}
\begin{tikzpicture}[scale=1,auto,swap,level distance=8mm]
\node[grayKnot] (B0) {$B_0$} 
  child {node {$f_1$}
    child {node[grayKnot] (B1) {$B_1$}
      child {node {$x_1$}}
    }
    child {node (B2) {$B_2$}
      child {node (C0) {$C_0$}
        child {node {$f_2$}
          child {node[grayKnot] (C1) {$C_1$}
            child {node {$x_2$}}
          }
          child {node[grayKnot] (C2) {$C_2$}
            child {node {$x_3$}}
          }
        }
      }
    }
  }
;

\fill[gray,opacity=0.3] \convexpath{B2,C0}{10pt};

\node [left = 15pt of B0] (a0) {$A_0$};
\draw [->] (a0) -- ++ (20pt,0);
\node [left = 15pt of B1] (a1) {$A_1$};
\draw [->] (a1) -- ++ (20pt,0);
\node [right = 15pt of B2] (a2) {$A_2$};
\draw [->] (a2) -- ++ (-20pt,0);
\node [left = 15pt of C1] (a3) {$A_3$};
\draw [->] (a3) -- ++ (20pt,0);
\node [right = 15pt of C2] (a4) {$A_4$};
\draw [->] (a4) -- ++ (-20pt,0);
\end{tikzpicture}
\hspace*{\fill}
\caption{The skeleton tree $t_A(x_1,x_2,x_3)$ and the tree $t_B(x_1, t_C(x_2,x_3))$}
 \label{fig-rewrite-skeleton-1}
\end{figure}
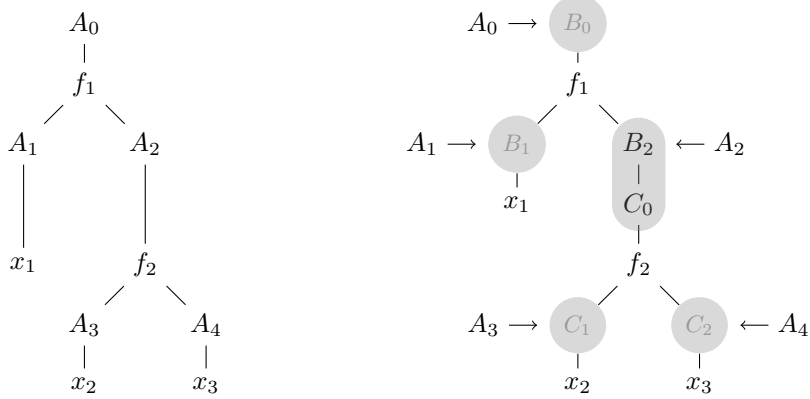

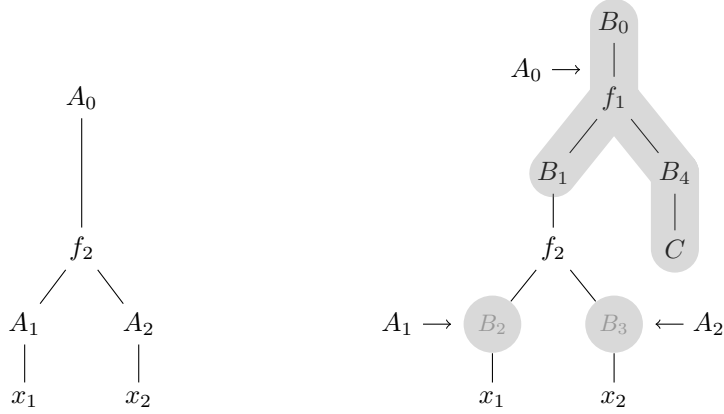
\begin{figure}[t]
\tikzset{level 1/.style={sibling distance=32mm}}
\tikzset{level 2/.style={sibling distance=16mm}}
\tikzset{level 7/.style={sibling distance=8mm}}  
\tikzset{level 8/.style={sibling distance=4mm}}
\hspace*{\fill}
\begin{tikzpicture}[scale=1,auto,swap, level 1/.style={level distance=20mm}, level 2/.style={level distance=10mm}]
\node (eps) {$A_0$} 
  child {node {$f_2$}
    child {node {$A_1$}
      child {node {$x_1$}}  
    }
    child {node {$A_2$}
      child {node {$x_2$}
      }
    }
  }
;
\end{tikzpicture}
\hspace*{\fill}
\begin{tikzpicture}[scale=1,auto,swap,level distance=10mm]
\node (B0) {$B_0$} 
  child {node (f1) {$f_1$}
    child {node (B1) {$B_1$}
      child {node {$f_2$}
        child {node[grayKnot] (B2) {$B_2$}
          child {node {$x_1$}}
        }
        child {node[grayKnot] (B3) {$B_3$}
          child {node {$x_2$}}
        }
      }
    }
    child {node (B4) {$B_4$}
      child {node (C) {$C$}}  
    }
  }
;
\fill[gray,opacity=0.3] \convexpath{B1,f1,B0,f1,B4,C,B4,f1}{9pt};
\node [above left = -5pt and 15pt of f1] (a0) {$A_0$};
\node [left = 15pt of B2] (a1) {$A_1$};
\node [right = 15pt of B3] (a2) {$A_2$};
\draw [->] (a0) -- ++(20pt,0);
\draw [->] (a1) -- ++(20pt,0);
\draw [->] (a2) -- ++(-20pt,0);
\end{tikzpicture}
\hspace*{\fill}
\caption{The skeleton tree $t_A(x_1,x_2)$ and the tree $t_B(x_1, x_2, t_C)$}
 \label{fig-rewrite-skeleton-2}
\end{figure}


Going from a monadic TSLP to a circuit (or dag) that evaluates over every semiring to the same 
noncommutative polynomial is easy:

\begin{lemma} \label{lemma-TSLP-circuit}
From a given monadic TSLP $\mc G$ over the terminal alphabet $\Sigma_m$ such that all productions are
of the form shown in Lemma~\ref{lemma-monadic-logspace},
one can construct in logspace (linear time, respectively)
an arithmetical circuit $C$ of depth $O(\mathrm{depth}(\mc G))$ and size $O(|\mc G|)$
such that over every semiring, $C$ and $\val(\mc G)$ evaluate to the same noncommutative polynomial in $m$ variables.
\end{lemma}

\begin{proof}
Fix an arbitrary semiring $\mc S$ and 
let $\mc R$ be the polynomial semiring $\mc R = \mc S[y_1, \ldots, y_m]$.
Clearly, for a nonterminal $A$ of rank $0$, $\val_{\mc G}(A)$ is a tree without
parameters that evaluates to an element $p_A$ of the semiring $\mc R$. For a 
nonterminal $A(x)$ of rank $1$, $\val_{\mc G}(A)$ is a tree, in which the only parameter $x$ occurs exactly
once.  Such a tree evaluates to a noncommutative polynomial $p_A(x) \in 
\mc R[x]$.  Since the parameter $x$ occurs exactly once in the tree $\val(A)$, it turns out
that $p_A(x)$ is linear and contains exactly one occurrence of $x$. More precisely, by induction on the structure 
of the TSLP $\mc G$ we show that for every nonterminal $A(x)$ of rank $1$, the tree $\val_{\mc G}(A)$ evaluates
in $\mc R[x]$ to a noncommutative polynomial of the form
$$
p_A(x) = A_0 + A_1 x A_2,
$$ where 
$A_0, A_1, A_2 \in \mc R = \mc S[y_1, \ldots, y_m]$.  Using the same induction, one can
build up a circuit of size $O(|\mc G|)$ and depth
$O(\mathrm{depth}(\mc G))$ that contains gates evaluating to
$A_0, A_1, A_2$. For a nonterminal $A$ of rank zero, the circuit contains
a gate that evaluates to the semiring element $p_A \in \mc R$, and we denote this gate
with $A$ as well.

The induction uses a straightforward case distinction on the rule for $A(x)$.
The cases that the unique rule for $A$ has the form $A(x) \to x$, $A(x) \to B(x)$,
$A(x) \to f(x,B)$, $A(x) \to f(B,x)$, $A \to f(B,C)$, or $A \to a$ 
is clear. For instance, for a rule $A(x) \to +(B,x)$, we have $p_A(x) = B + 1 \cdot x \cdot 1$,
i.e., we set $A_0 := B$, $A_1 := 1$, $A_2 := 1$.
Now consider a rule $A(x) \to B(C(x))$
(for $A \to B(C)$ the argument is similar). We have already built up a circuit containing
gates that evaluate to $B_0, B_1,B_2, C_0, C_1, C_2$, where
$$
p_B(x)  =  B_0 + B_1 x_1 B_2, \qquad p_C(x)  =  C_0 + C_1 x C_2  .
$$
We get
\begin{eqnarray*}
p_A(x) &=& p_B(p_C(x)) \\
&=& B_0 + B_1 (C_0+C_1 x C_2)  B_2  \\
&=& (B_0 + B_1 C_0 B_2) + B_1 C_1 x_1 C_2 B_2 
\end{eqnarray*}
and therefore set
\begin{equation*}
A_0 := B_0 + B_1 C_0 B_2, \ A_1 := B_1 C_1, \  A_2 := C_2 B_2 .
\end{equation*}
So we can define the polynomials $A_0, A_1, A_2$ using the gates $B_0$, $B_1$, $B_2$,
$C_0$, $C_1$, $C_2$ with only $5$ additional gates. 
Note that also the depth only increases by a constant factor (in fact, 2).

The output  gate of the circuit is the start nonterminal  of the TSLP $\mc G$.
The above construction can be carried out in linear time as well as in logspace.
\end{proof}
Now we can show the main result of this section:

\begin{theorem} \label{thm-circuit}
A given arithmetical formula $F$ of size $n$ having $m$ different variables
can be transformed in logspace (linear time, respectively)
into an arithmetical circuit $C$ of depth $O(\log n)$ and size $O\big(\frac{n \cdot \log m}{\log n}\big)$
such that over every semiring, $C$ and $F$ evaluate to the same noncommutative polynomial in $m$ variables.
\end{theorem}

\begin{proof}
Let $F$ be an arithmetical formula of size $n$ and let $y_1, \ldots, y_m$ be the variables occurring in $F$.
Fix an arbitrary semiring $\mc S$ and 
let $\mc R$ be the polynomial semiring $\mc R = \mc S[y_1, \ldots, y_m]$.
Using Lemma~\ref{lemma-monadic-logspace} we can construct in logspace (linear time, respectively)
a monadic TSLP $\mc G$ of size $O\big(\frac{n \cdot \log m}{\log n}\big)$
and depth $O(\log n)$ such that $\val(\mc G) = F$.
Finally, we apply Lemma~\ref{lemma-TSLP-circuit} in order to transform $\mc G$ 
in logspace (linear time, respectively)
into an equivalent circuit 
of size $O\big(\frac{n \cdot \log m}{\log n}\big)$ and depth $O(\log n)$.
\end{proof}

Theorem~\ref{thm-circuit} can also be shown for fields instead of semirings. In this case, the expression
is built up using variables, the constants $-1$, $0$, $1$, and the field operations $+, \cdot$ and $/$. The proof is similar to the semiring case. 
Again, we start with a monadic TSLP of size $O\big(\frac{n \cdot \log m}{\log n}\big)$ and depth $O(\log n)$ for the arithmetical
expression. Again, one can assume that all rules have the form $A(x) \to B(C(x))$, $A \to B(C)$,
$A(x) \to f(x,B)$, $A(x) \to f(B,x)$, $A \to f(B,C)$, $A(x) \to B(x)$, $A(x) \to x$, or $A \to a$, where $f$ is one of the binary
field operations and $a$ is either $-1$, $0$, $1$, or a variable. Using this particular rule format, one can
show that every nonterminal $A(x)$ evaluates to a rational function $(A_0 + A_1 x)/(A_2 + A_3 x)$ for polynomials
$A_0, A_1, A_2, A_3$ in the circuit variables, whereas a nonterminal of rank $0$ evaluates to a fraction of two polynomials.
Finally, these polynomials can be computed by a single circuit of size $O\big(\frac{n \cdot \log m}{\log n}\big)$ and depth $O(\log n)$.

Lemma~\ref{lemma-TSLP-circuit} has an interesting application to the problem of checking whether the polynomial represented by a
TSLP over a ring is the zero polynomial. The question, whether the polynomial computed by a given circuit is the zero polynomial
is known as {\em polynomial identity testing} (PIT). This is a famous problem in algebraic complexity theory. For the case that
the underlying ring is $\mathbb{Z}$ or $\mathbb{Z}_n$ ($n \geq 2$) polynomial identity testing 
belongs to the complexity class {\bf coRP} (the complement of randomized polynomial time), see \cite{AgraSap09}.
PIT can be generalized to arithmetic expressions that are given by a TSLP. Using Lemma~\ref{lemma-TSLP-circuit}  and Theorem~\ref{theo-monadic}
we obtain:

\begin{theorem}
Over any semiring, the question, whether the polynomial computed by
a given TSLP is the zero polynomial, is equivalent with respect to polynomial time reductions to PIT.
In particular, if the underlying semiring is $\mathbb{Z}$ or $\mathbb{Z}_n$, then the question, whether the polynomial computed by
a given TSLP is the zero polynomial,  belongs to {\bf coRP}.
\end{theorem}

 \section{Future work}
 
 In \cite{ZhangYK14} a universal (in the information-theoretic sense) code for binary trees is developed. This code 
 is computed in two phases: In a first step, the minimal dag for the input tree is constructed. Then, a particular
 binary encoding is applied to the dag. It is shown that the {\em average redundancy} of the resulting code converges
 to zero (see \cite{ZhangYK14} for definitions) for every probability distribution on binary trees that satisfies the so-called domination property and 
 the representation ratio negligibility property. Whereas the domination property is somewhat technical and 
 easily  established for many distributions, the representation ratio negligibility property means that
 the average size of the dag divided by the tree size converges to zero for the underlying probability distribution.
 This is, for instance, the case for the uniform distribution, since the average size of the dag is 
 $\Theta\big(\frac{n}{\sqrt{\log n}}\big)$ \cite{FlajoletSS90}.   
 
 We construct TSLPs that have {\em worst-case} size $O\big(\frac{n}{\log n}\big)$ assuming a constant number of node labels. 
 We are confident that replacing the minimal dag by a TSLP of worst-case size $O\big(\frac{n}{\log n}\big)$ in the 
 universal encoder from \cite{ZhangYK14} leads to stronger
 results. In particular, we hope that for the resulting encoder the {\em maximal pointwise redundancy} converges
 to zero for certain probability distributions. 
 For strings, such a result was obtained in \cite{KiYa00} using the fact that every string of length $n$ over a constant size alphabet has 
 an SLP of size $O\big(\frac{n}{\log n}\big)$.
 
It would be interesting to know  the worst-case output size of the grammar-based tree compressor 
 from \cite{JezLo14approx}. It works in linear time and produces a TSLP that is only by a factor $O(\log n)$ larger than an optimal TSLP.
 From a complexity theoretic point of view, it would be also interesting to see, whether {\sf TreeBiSection} (which can be implemented in logspace)
 can be even carried out in $\mathsf{NC}^1$. 
 
In \cite{BilleGLW13} the authors proved that the top dag of a given tree $t$ of size $n$ is at most by a factor $\log n$  
larger than the minimal dag of $t$. It is not clear, whether the TSLP constructed by {\sf TreeBiSection} has this property too.
The construction of the top dag is done in a bottom-up way, and as a consequence identical subtrees are compressed
in the same way. This property is crucial for the comparison with the minimal dag. {\sf TreeBiSection}  works in a top-down
way. Hence, it is not guaranteed that  identical subtrees are compressed in the same way. In contrast, {\sf BU-Shrink} works 
bottom-up, and one might compare the size of the produces TSLP with the size of the minimal dag.

\medskip
\noindent
{\bf Acknowledgment.} We have to thank Anna G{\'a}l  for helpful comments.

\def\cprime{$'$} \def\cprime{$'$}

%

\end{document}